\documentclass{amsart}

\newtheorem{theorem}{Theorem}[section]
\newtheorem{lemma}[theorem]{Lemma}
\newtheorem{proposition}{Proposition}[section]
\usepackage{amsmath,amssymb,amsthm}
\theoremstyle{definition}
\newtheorem{definition}[theorem]{Definition}
\newtheorem{example}[theorem]{Example}

\newtheorem{corollary}[theorem]{Corollary}
\theoremstyle{remark}
\newtheorem{remark}[theorem]{Remark}
\usepackage{enumitem}
\numberwithin{equation}{section}
\usepackage{graphicx}%
\usepackage{multirow}%
\usepackage{amsmath,amssymb,amsfonts}%
\usepackage{amsthm}%
\usepackage{mathrsfs}%
\usepackage[title]{appendix}%
\usepackage{xcolor}%
\usepackage{textcomp}%
\usepackage{manyfoot}%
\usepackage{booktabs}%
\usepackage{algorithm}%
\usepackage{algorithmicx}%
\usepackage{algpseudocode}%
\usepackage{listings}%
\usepackage{hyperref} 
\begin{document}
	\title[Weight Distribution, $I$-perfect and MDS Codes]{Weight Distribution of the Weighted Coordinates  Poset Block Space and Singleton Bound}
	\author{Atul Kumar Shriwastva}
	\address{{Department of Mathematics, National Institute of Technology Warangal,  Telangana 506004, India}}
	\email{shriwastvaatul@student.nitw.ac.in}
	\thanks{}
	
	\author{R S Selvaraj}
	\address{{Department of Mathematics, National Institute of Technology Warangal,  Telangana 506004, India}}
	\email{rsselva@nitw.ac.in}
	\thanks{}
	
	\subjclass[2010]{Primary: 94B05, 15A03; Secondary: 06A06}
	
	\keywords{Weight distribution,  Poset block codes, Weighted coordinates poset codes, NRT block codes, Hierarchical poset space, Singleton bound, MDS codes}
	
	\date{}
	
	\dedicatory{}

\begin{abstract}
    In this paper,
	we determine the complete weight distribution of the space  $ \mathbb{F}_q^N $ endowed by the weighted coordinates poset block metric ($(P,w,\pi)$-metric), also known as the $(P,w,\pi)$-space, thereby obtaining it for $(P,w)$-space, $(P,\pi)$-space, $\pi$-space, and $P$-space as special cases. Further, when $P$ is a chain, the resulting space is called as Niederreiter-Rosenbloom-Tsfasman (NRT) weighted  block space and when $P$ is hierarchical, the resulting space is called as weighted coordinates hierarchical poset block space. The complete weight distribution of both the spaces are deduced from the main result. Moreover, we define an $I$-ball for an ideal $I$ in $P$ and study the characteristics of it in $(P,w,\pi)$-space. 
	We investigate the relationship between the $I$-perfect codes and $t$-perfect codes in $(P,w,\pi)$-space. Given an ideal $I$, we investigate how the maximum distance separability (MDS) is related with $I$-perfect  codes and $t$-perfect codes in $(P,w,\pi)$-space. Duality theorem is derived for an MDS $(P,w,\pi)$-code when all the blocks are of same length.	Finally, the  distribution of codewords among $r$-balls is analyzed in the case of chain poset, when all the blocks are of same length.
\end{abstract}

	\maketitle

\section{Introduction}
{T}{he} classical problem of coding theory \cite{hnt} is to find a linear  code $[n,k,d]$ with the largest minimum distance $d$ of length $n$ over the finite field $\mathbb{F}_q$, for any integer $ n > k \geq 1$.  Several researchers explored this problem with various metrics such as Lee metric \cite{Leecode}, rank metric\cite{rankmetric}, poset metric \cite{Bru, crownauto, wdcrown},  block metrics \cite{fxh, Ebc}, pomset metric \cite{gsrs} and weighted coordinates poset metric \cite{wcps, aks}.
Poset metric is introduced by Brualdi et al. \cite{Bru} on $\mathbb{F}_q^n$ with the help of partially ordered relations on the set $[n]$, where $[n] = \{1,2,\ldots,n\}$ represents the coordinate positions of $n$-tuples in the vector space $\mathbb{F}_q^n$.  Hyun and Kim \cite{hkmdspc} expanded the work of studying perfect codes  by introducing the  $J$-ball associated with an ideal $J$ in poset $P$  and by studying the $J$-perfect codes. Before their work, research on perfect codes was primarily concerned with $r$-balls \cite{crown, Ch1-LeePerfectPoset2004,Ch1-LeePerfectPoset2006}, but the introduction of $J$-balls provided a new perspective that has been extensively studied by subsequent authors in \cite{bkdnsr, aks, gr}. 
  A space equipped with pomset mertics is introduced by I. G. Sudha and R. S. Selvaraj \cite{gsrs} which is a generalization of Lee space \cite{Leecode} and  poset space \cite{Bru} over $\mathbb{Z}_m$. 
 In \cite{wcps}, L. Panek \cite{wcps} introduced the weighted coordinates poset metric which is a simplified version of the pomset metric  that does not use the multiset structure.
  In that paper, the complete description of the group of linear
 isometries of $(P,w)$-space was described for an arbitrary poset $P$ and the Singleton bound was established for codes with chain poset. Recently, in \cite{aks}, we determined the complete weight enumerator of the space $\mathbb{F}_q^n$ endowed by the $(P,w)$-metric and the Singleton bound is established for codes  with any arbitrary poset $P$.  
  
\par During the past two decades, research into block codes has driven notable progress in the communication field which includes significant developments in experimental design, high-dimensional numerical integration, and cryptography. As a result, block codes have become an essential area of study within digital communication.
Block codes of length $N$ in $ \mathbb{F}_{q}^N$ which was at first introduced by K. Feng \cite{fxh} with the help of  a label map $\pi$ from $[n] $ to $ \mathbb{N}$ such that $\sum_{i=1}^{n}\pi (i) = N$ and by considering $\mathbb{F}_{q}^{N} = \mathbb{F}_{q}^{\pi(1)}  \oplus \mathbb{F}_{q}^{\pi(2)} \oplus \ldots \oplus \mathbb{F}_{q}^{\pi(n)}$, are later  extended by M. M. S. Alves et al. \cite{Ebc} by using partial order relation on the block  positions $[n]$ that are  known as poset block or $(P,\pi)$-block codes.  Codes endowed with the Niederreiter-Rosenbloom-Tsfasman block metric (NRT block metric) are a particular case of poset block metric when the poset is a chain, have been classified in \cite{nrt block,  MacAdmittingclassification}. 
B. K. Dass et. al. \cite{bkdnsr} extended the $J$-perfect codes to $J$-perfect poset block codes for an ideal $J$ in the poset ${P}$.
The  characteristics of those codes such as minimum distances, covering radius, packing radius, Singleton bound, maximum distance separability, diameter,  perfectness and weight distribution were being studied comprehensively (see  \cite{bk, bkdn, bkd, bkdnsr, nrt block, aksh}). Due to the large number of applications of block codes, the authors are motivated to extend the pomset code to  pomset block codes \cite{wmjl,aksh,akshr} and weighted coordinates poset codes to the weighted coordinates poset block codes \cite{as,aks}. 
In \cite{as}, we defined the weighted coordinates poset block metric ($d_{(P,w,\pi)}$) on the space $\mathbb{F}_q^N$ which generalizes several metrics such as Hamming metric, poset metric,  $(P,w)$-metric, pomset metric, $(P,\pi)$-metric and pomset block metric. 
In parallel to these,   similar works on weighted coordinates poset block metric were investigated in \cite{wj, wjl}.
\par In \cite{wdcrown}, D. S. Kim and S. H. Cho derived a formula that gave the weight distribution of space $ \mathbb{F}_{q}^n$ equipped with a particular class of poset structure called crown poset by first finding the MacWilliams type identities on codes over such space. In this paper, we determine complete weight distribution of  $\mathbb{F}_{q}^N$ equipped with weighted coordinates poset block structure by suitable combinatorial tools, as the blocks are weighted by a weight function $w$ on  $\mathbb{F}_{q}$. Moreover, we study the $I$-balls, $r$-balls, MDS codes and their perfectness for $(P,w,\pi)$-codes over $\mathbb{F}_q$ with arbitrary poset $P$.
\par  First, we introduce weighted coordinates poset block metric (or  $(P,w,\pi)$-metric) for codes of length $N$ over the field $\mathbb{F}_q$. Section \ref{section:2} provides the necessary preliminaries. We will  define the $(P,w,\pi)$-metric in Section \ref{section:3}, by attaching a weight $\tilde{w}^{k_i}$ to each block of length $k_i$ of a codeword where $\tilde{w}^{k_i}$ is a weight on $\mathbb{F}_q^{k_i} $ that depends on a weight  $w$ on $\mathbb{F}_q$.  In Section \ref{section:4}, we will determine the complete weight distribution of $\mathbb{F}_q ^{N}$  with respect to $(P,w,\pi)$-metric, in general, thereby, facilitating the weight distribution  with respect to  $(P,w)$-metric, $(P,\pi)$-metric, $\pi$-metric, $P$-metric as particular cases. In one sense, the results obtained are generalizations of those established by M. M. Skriganov in \cite{skri}. In the sequel, we determine the  cardinality of an $r$-ball with respect to $(P,w,\pi)$-metric and  $(P,w)$-metric. 
In Section \ref{section:5}, the complete weight distribution for hierarchical weighted block spaces are determined by considering $P$ to be a hierarchical poset.
In Section \ref{perfect codes in P,w,pi-space},  $I$-perfect codes for an ideal $I$ in $P$ are introduced and the relationship between $I$-perfect codes and $t$-perfect codes with respect to $(P,w,\pi)$-metrics is investigated. Further, in section \ref{MDS-P,w,pi-space},  we find the relation between MDS and $I$-perfect codes when all the blocks are of same length. Moreover, we establish the duality theorem for those $(P,w, \pi)$-codes when all the blocks are of same length.  In Section \ref{NRT-P,w,pi-space},  the complete weight distribution
of the $(P,w,\pi)$-space when $P$ is a chain is deduced. Further, 
the distribution of codewords among $I$-balls is also analyzed when the poset is a chain. Moreover, it concludes with the establishing of weight distribution of MDS codes in  $(P,w, \pi)$-space when the poset is a chain.

\section{Preliminaries}
\label{section:2}
 	For a  positive integer $m$, the  Lee weight $w_L$ of $a\in \mathbb{Z}_m$ is $\min\{a, m-a\}$  and the Lee weight of an $n$-tuple $u=(u_1,u_2,\ldots,u_n) \in \mathbb{Z}_m^n$ is $w_L{(x)} = \sum_{i=1}^{n}{w_L{(u_i)}} $. Lee distance between  $u, v \in \mathbb{Z}_{m}^n$ is  $ d_{L}(u,v) = w_L{(u - v)} $. Support of an $n$-tuple $u=(u_1,u_2,\ldots,u_n) \in \mathbb{F}_q^n$ is defined to be the set $ supp(u)= \{i \in [n] : u_i \neq 0  \}$ and the Hamming weight of $u$ is $w_H(u)= |supp(u)|$. 
\par Let $P=([n],\preceq)$ be a poset.  An element $j \in A \subseteq P$ is said to be a maximal element of $A$ if there is no $i \in A$ such that $j \preceq i$. An element $j \in A \subseteq P$ is said to be a minimal element of $A$ if there is no $i \in A$ such that $i \preceq j$.  A subset $I$ of $P$ is said to be an ideal if $j \in I$ and $i \preceq j$ imply $i \in I$. For a subset $A$  of $P$, an ideal generated by $A$ is the smallest ideal containing $A$ and is denoted by $\langle A \rangle$.  
Poset weight or $P$-weight of $u \in \mathbb{F}_q^n$   is $ w_P (u)= |\langle supp(u) \rangle|$ and $P$-distance between $u,v \in \mathbb{F}_q^n$ is $ d_{P} (u,v)=  w_{P}   (u-v)$.  
\par Through a label map  $\pi  : [n] \rightarrow \mathbb{N} $ defined as $\pi (i) =k_i $ such that $ \sum\limits_{i=1}^{n} \pi (i) = N  $ and considering  $ \mathbb{F}_q^N $ as the direct sum  $ \mathbb{F}_{q}^{k_1} \oplus \mathbb{F}_{q}^{k_2} \oplus \cdots \oplus \mathbb{F}_{q}^{k_n} $, one can express $x \in \mathbb{F}_q^N$  uniquely as  $x= x_1 \oplus x_2 \oplus \cdots \oplus x_n $ with  $x_i = (x_{i_1},x_{i_2},\ldots,x_{i_{k_i}}) \in \mathbb{F}_{q}^{k_i} $ and can define the block-support  \cite{Ebc},  \cite{fxh} or $\pi$-support of $x$ as $ supp_\pi(x)= \{i \in [n] : 0 \neq  x_i  \in \mathbb{F}_q^{k_i}\}$. 
Now, $\pi$-weight of $x$ is defined as $ w_\pi (x)= |supp_\pi(x)|$ and $\pi$-distance 
between  $x,y \in \mathbb{F}_q^N$ is $ d_\pi (x,y)=  w_\pi (x-y)$. Note that if $k_i =1 ~ \forall ~ i$, the $\pi$-weight and $\pi$-distance becomes Hamming weight and Hamming distance respectively. Now, the poset block weight or  $ (P,\pi)$-weight of $x$ is defined as $ w_{(P,\pi)} (x) \triangleq  |\langle supp_\pi(x) \rangle|$ and $(P,\pi)$-distance between $x,y \in \mathbb{F}_q^N$ is defined as $ d_{(P,\pi)} (x,y) \triangleq  w_{(P,\pi)}   (x-y)$. 
\par By assigning a weight for each coordinate position of a vector  $u \in \mathbb{F}_{q}^{n} $, L. Panek \cite{wcps} introduced what is called as  weighted coordinates poset weight. For this, let $u \in \mathbb{F}_{q}^{n} $, 
$I_{u}^P = \langle supp(u) \rangle $  and $M_{u}^P$ be the set of all maximal elements in $I_{u}^P$.  If $w$ is  a weight defined on $\mathbb{F}_q$ and $M_w = \max \{w(\alpha) :   \alpha \in \mathbb{F}_q \}$, then the weighted coordinates poset weight or $(P,w)$-weight of  $u $  is  defined as 
$ 	w_{(P,w)}(u) = \sum\limits_{i \in M_{u}^P} w{(u_i)} + \sum\limits_{i \in {I_{u}^{P} \setminus M_{u}^P}} M_w $. The weighted coordinates poset distance or 
$(P,w)$-distance between   $u,v \in \mathbb{F}_q ^n$ is defined as $ d_{(P,w)}(u,v) \triangleq w_{(P,w)}(u-v) $. If $w(\alpha) = 1$ $\forall$  $ \alpha \in \mathbb{F}_q \setminus 0$, then the $(P,w)$-weight of  $u $  becomes the $P$-weight of  $u $.
\par
\par Throughout  the paper, let $ \mathcal{I}(P) = \{ I \subseteq P: I  \text{~is an ideal} \}$ denote the collection of all ideals in $P$.  Let $\mathcal{I}^{i}$ be its sub-collection of all ideals whose cardinality is $i$. Let 
$\mathcal{I}_{j}^{i}$ denote the collection of all ideals $I \in \mathcal{I}(P)$ with cardinality $i$ having $j$ maximal elements. Clearly, $\mathcal{I}_{j}^{i} \subseteq \mathcal{I}^{i}$ and $\cup_{j=1}^{i} \mathcal{I}_{j}^{i} =\mathcal{I}^{i} $, for $1 \leq i \leq n$. Given an ideal $I \in \mathcal{I}_{j}^{i}$, let $Max(I) = \{i_1, i_2, \ldots, i_j\}$ be the set of all maximal elements of $I$ and $I \setminus Max(I) = \{ l_1, l_2,\ldots, l_{i-j}\}$ be the set of all non-maximal elements of $I$, if any. Here $i,j$ are integers such that $ 1 \leq i \leq n $ and $ 1 \leq j \leq i$. 
\begin{proposition}\label{pwpiChept-I-J in I will exist}
	Let $P = ([n], \preceq)$ be a poset. Let $0 \leq s \leq t \leq n$. Then  for each $I \in \mathcal{I}^t$ there exists $J \in \mathcal{I}^s$ such that $J \subseteq I$. Moreover, for each $J \in \mathcal{I}^s$ there exists $I \in \mathcal{I}^t$ such that $J \subseteq I$.
\end{proposition}  
\begin{theorem}[\cite{aks}]\label{pwpiChept-I-J insdie I minusMaxI}
	Suppose that $P$ has exactly one ideal $J$ with cardinality $t$, where $t \leq n-1$. Then, the following four statements hold true: 	
	\begin{enumerate}[label=(\roman*)]
		\item\label{rev0} $J \subset I$ for any ideal $I$ with $| I | > t$;
		\item\label{rev1}  $J \subseteq I \setminus Max(I)$ for every ideal $I$ with $| I | > t$;
		\item\label{rev2} For any $a \in J$ and $b \in [1,n] \setminus J$, it holds that $a \preceq b$ in $P$;
		\item\label{rev3} For every ideal $I$ with  $I \not\subseteq J$, it holds that  $J \subsetneqq I$ and $J \subseteq I \setminus Max(I)$.		
	\end{enumerate}
\end{theorem}

\section{Weighted Coordinates Poset Block Metrics}
\label{section:3}
To make this paper a self-contained one and to proceed smoothly into the findings given in further sections, we re-introduce\footnote{
	A. K. Shriwastva and R. S. Selvaraj,	{Weighted coordinates poset block codes},\textit{Discrete Math. Algorithms Appl},	communicated (April 2023) (unpublished). \\
	In this paper, the notion of weighted coordinates poset block metrics is introduced.} 
the content of the section from \cite{as}.  
\par
To start with, let  $w$ be a weight on $\mathbb{F}_q$, $m_w = \min \{w(\alpha) :   \alpha \in \mathbb{F}_q, \alpha \neq 0\}$ and $M_w = \max \{w(\alpha) :   \alpha \in \mathbb{F}_q \}$. For a $k \in \mathbb{N}$, and a $v=(v_1, v_2,\ldots,v_k) \in  \mathbb{F}_q^{k} $, we define  $\tilde{w}^{k}{(v)} = \max \{ w(v_i) : 1 \leq i \leq k\}$. Clearly,   $\tilde{w}^{k}$ is a weight on $\mathbb{F}_q^{k} $ induced by the weight $w$. On 	$ \mathbb{F}_{q}^{k_i} $, $1 \leq i \leq n$, we call   $\tilde{w}^{k_i}$, a block weight. Note that  $\tilde{w}^{k}$ is not an additive weight and the metric $d_{\tilde{w}^{k}}$ induced by ${w}$-weight on $\mathbb{F}_q^{k} $ defined as $d_{\tilde{w}^{k}} (u,v) \triangleq \tilde{w}^{k}(u-v)$  is not an additive metric.
\par
Considering the finite set $[n]=\{1,2,\ldots,n\}$ with $\preceq $ to be a partial order,  the pair $P=([n],\preceq)$ is  a poset.  With a label map  $\pi  : [n] \rightarrow \mathbb{N} $ defined  as $\pi (i) =k_i $ in the previous section such that $ \sum\limits_{i=1}^{n} \pi (i) = N  $, a positive inetger, we have $	\mathbb{F}_{q}^{N} = \mathbb{F}_{q}^{k_1}  \oplus \mathbb{F}_{q}^{k_2} \oplus \ldots \oplus \mathbb{F}_{q}^{k_n} $. Thus, if  $x \in  \mathbb{F}_{q}^{N} $ then $x= x_1 \oplus x_2 \oplus \cdots \oplus x_n $ with  $x_i = (x_{i_1},x_{i_2},\ldots,x_{i_{k_i}}) \in \mathbb{F}_{q}^{k_i} $. Let $I_{x}^{P,\pi} =  \langle supp_{\pi}(x) \rangle $ be the ideal generated by the $\pi$-support of $x$ and $M_{x}^{P,\pi}$ be the set of all maximal elements in $I_{x}^{P,\pi}$.
\begin{definition}[$(P,w,\pi)$-weight \cite{as}]
	\sloppy{Given a poset $P=([n],\preceq)$, a weight $w$ on $\mathbb{F}_q $, a label map  $\pi $ with a block weight $\tilde{w}^{k_i}$ on  $\mathbb{F}_q^{k_i} $ induced by $w$, the 
		weighted coordinates poset block weight or 
		$(P,w,\pi)$-weight of  $x \in \mathbb{F}_q^N$ is defined as
		\begin{align*}
			w_{(P,w,\pi)}(x) \triangleq \sum\limits_{i \in M_{x}^{P,\pi}} \tilde{w}^{k_i}{(x_i)} + \sum\limits_{i \in {I_{x}^{P,\pi} \setminus M_{x}^{P,\pi}}} M_w
		\end{align*} 
		The  $(P,w,\pi)$-distance between two vectors  $x,y \in \mathbb{F}_q ^N$ is defined as:
		$ d_{(P,w,\pi)}(x,y) \triangleq  w_{(P,w,\pi)}(x-y)$.}
\end{definition} 
\begin{theorem}[\cite{wj},\cite{as}]\label{t1}
	The $(P,w,\pi) $-distance is  a metric on $\mathbb{F}_{q}^N$.
\end{theorem}  

Here, $d_{(P,w,\pi)}$ defines a metric on $ \mathbb{F}_{q}^N $ called as \textit{weighted cordinates poset block metric} or $ (P,w,\pi) $-metric. The pair $ (\mathbb{F}_{q}^N,~d_{(P,w,\pi)} )$ is said to be a  $(P,w,\pi)$-space.
\par  A $ (P,w,\pi) $-block code  $\mathcal{C} $ of length $N$ is a subset of $ (\mathbb{F}_{q}^N,~d_{(P,w,\pi)} )$-space  and 
\begin{align*}
	d_{(P,w,\pi)}(\mathcal{C}) = \min \{  d_{(P, w,\pi)} {(c_1, c_2)}: c_1, c_2 \in \mathcal{C}, 
	c_1\neq c_2  \} 
\end{align*}
gives the minimum distance of  $\mathcal{C}$. 
If $\mathcal{C}$ is a linear $ (P,w,\pi) $-block code, then   
\begin{align*}
	d_{(P,w,\pi)}(\mathcal{C}) = \min \{ w_{(P,w,\pi)}(c) : 0 \neq c \in \mathcal{C} \} .
\end{align*}
As $ w_{(P,w,\pi)}(v) \leq n M_w$ for any $v \in \mathbb{F}_{q}^N $, the minimum distance of a linear code   $\mathcal{C} $ is bounded above by $ n M_w $.
\par Though it is a more general fact to say that  the $(P,w,\pi) $-distance is  a metric on a free module $R^N$ over a commutative ring $R$ with identity, we focus our discussion by taking $R = \mathbb{F}_{q}$ almost everywhere. The results obtained in the forthcoming sections for 
$\mathbb{F}_{q}^N $ or $\mathbb{Z}_{m}^N $ are even valid for any  free module $R^N$ over the  commutative ring $R$ with identity.
\par We remark that the $(P,w,\pi)$-distance is a metric which includes several classic metrics of coding theory.

\begin{remark}[\cite{as}]\label{becomes} 
	Now we shall see how or in what manner the weighted coordinates poset block metric generalizes the various poset metrics arrived at so for.
	\begin{enumerate}[label=(\roman*)]
		\item If $k_i = 1 $ for every $i\in [n]$ then the block support of $v \in \mathbb{F}_q^N$ becomes the usual support of $v \in \mathbb{F}_q^n$ 
		and $(P,w,\pi)$-weight of $v \in \mathbb{F}_q^N$ becomes $(P,w)$-weight of $v \in \mathbb{F}_q^n$. Thus, the $(P,w,\pi)$-space becomes the $(P,w)$-space (as in  \cite{wcps}).
		\item 	If $k_i = 1 $ for every $i\in [n]$ and $w$ is the Hamming weight on $\mathbb{F}_q$ so that  $w(\alpha)=1$ for $0 \neq \alpha \in \mathbb{F}_q$,  then for any  
		$ v = v_1 \oplus v_2 \oplus \ldots \oplus v_n \in \mathbb{F}_q^N$, 
		\begin{align*}
			\tilde{w}^{k_i}{(v_i)} &=
			\max\{w(v_{i_t}) : 1 \leq t \leq k_i\} \\&= \left\{ 
			\begin{array}{ll}
				0, &  \text{if} \ v_{i} = 0\\
				1, &   \text{if} \ v_{i} \neq 0
			\end{array}	  \right\} \\&=  w_{H}(v_{i}),
		\end{align*}
		the Hamming weight of $v_i$ in $\mathbb{F}_q$.
		Thus, the	 $(P,w,\pi)$-space becomes the poset space or $P$-space (as in \cite{Bru}).
		\item 	If $w$ is the Hamming weight  on $\mathbb{F}_q$.  Let $v \in \mathbb{F}_q^N$, $v_i = (v_{i_1},v_{i_2}, \ldots,v_{i_{k_i}}) \in \mathbb{F}_{q}^{k_i} $ and $v_{i_{t}} \in  \mathbb{F}_{q}$, $1 \leq t \leq k_i$, so that  $ 	 \tilde{w}^{k_i}{(v_i)} = \max\{w(v_{i_t}) : 1 \leq t \leq k_i\}   = \left\{ 
		\begin{array}{ll}
			0, &  \text{if} \ v_{i} = 0\\
			1, &    \text{if} \ v_{i} \neq 0
		\end{array}	  \right\}  $. Hence,
		\begin{align*}
			w_{(P,w,\pi)}(v) &= | M_{v}^{P,\pi} |  + | {I_v}^{P,\pi} \setminus M_{v}^{P,\pi} | \\&= | I_{v}^{P,\pi} | \\&= w_{(P,\pi)}(v). 
		\end{align*} 
		Thus, the  $(P,w,\pi)$-space becomes the $(P,\pi)$-space  (as in \cite{Ebc}).
		\item 	If $w$ is the Hamming weight on $\mathbb{F}_q$ and $P$ is an antichain, then for any $v \in \mathbb{F}_q^N$, 
		$ 	 \tilde{w}^{k_i}{(v_i)} =  \left\{ 
		\begin{array}{ll}
			0, &  \text{if} \ v_{i} = 0\\
			1, &   \text{if} \ v_{i} \neq 0
		\end{array}	  \right\} $, $I_v^{P,\pi} = supp_{\pi} (v)$ and hence,
		\begin{align*}
			w_{(P,w,\pi)}(v) &= | M_{v}^{P,\pi} |  + | {I_v}^{P,\pi} \setminus M_{v}^{P,\pi} | \\&= | supp_{\pi}(v) | \\& 
			= w_{\pi}(v) .
		\end{align*}
		Thus, the  $(P,w,\pi)$-space becomes the $\pi$-space or $( \mathbb{F}_q^N,d_{\pi})$-space  (as in \cite{fxh}). 
	\end{enumerate} 
\end{remark}
\par  The $(P,w,\pi)$-ball centered at a point $y \in \mathbb{F}_q^N$ with radius $r$ is the set 
\begin{align*}
	B_{r}(y)=\{x \in \mathbb{F}_{q}^{N} : d_{(P,w,\pi)}(y,x) \leq r\}
\end{align*}
\par \sloppy The $(P,w,\pi)$-sphere centered at $y$  with radius $r$ is  $S_{r}(y)=\{x \in \mathbb{F}_{q}^{N} : d_{(P,w,\pi)}(y,x) = r\}$.
$B_r(y)$ and $S_r(y)$ are also called as $r$-ball and $r$-sphere respectively. As the $d_{(P,w,\pi)}$-metric is translation invariant, we have  $B_{r}(y) = y + B_{r}(0)$ where $0 \in \mathbb{F}_q^N$. Clearly, $ | B_{r}(0) | =  1 + \sum\limits_{t=1}^{r} | S_{t}(0) |$ and $ | S_{t}(0) |$ is equal to the number of $x \in \mathbb{F}_{q}^{N}$ such that $w_{(P,w,\pi)}(x)=t$.   Therefore, to determine the cardinality of an $r$-ball, first we shall  find the $(P,w,\pi)$-weight distribution of $\mathbb{F}_{q}^{N}$. 
\par Hereafter, $\mathcal{C} $ denote a $(P,w,\pi)$-code of length $N = k_1 + k_2 + \ldots + k_n$ over $\mathbb{F}_q$ with minimum distance $d_{(P,w,\pi)}(\mathcal{C})$. Let $r_{\tilde{\omega}} = \bigg\lfloor \frac{d_{(P,w,\pi)} (\mathcal{C})-m_w}{ M_w} \bigg\rfloor $ for any $(P,w,\pi)$-code $\mathcal{C}$.
\begin{definition}
	A $(P,w,\pi)$-code $\mathcal{C}  $ of length $N$ over $\mathbb{F}_q$ is said to be a \textit{maximum distance separable (MDS)} $(P,w,\pi)$-code if it attains its Singleton bound.
\end{definition} 
\begin{theorem}[Singleton Bound for $(P,w,\pi)$-block Code \cite{as}] \label{pwpiChept-I-sbwpic}
	Let $\mathcal{C}$ be a $(P,w,\pi)$-code of length $N = k_1 + k_2 + \ldots + k_n$ over $\mathbb{F}_q$ with minimum distance $d_{(P,w,\pi)}(\mathcal{C})$ and let $r_{\tilde{\omega}} = \bigg\lfloor \frac{d_{(P,w,\pi)} (\mathcal{C})-m_w}{ M_w} \bigg\rfloor $. Then    
	$
	\max\limits_{J \in  \mathcal{I}^{r_{\tilde{\omega}}}}    \bigg\{\sum_{i \in J} k_{i}\bigg\} \leq N - \lceil log_{q}|\mathcal{C}| \rceil .
	$
\end{theorem}
If $k_i = s $ $\forall$  $i \in  [n]$  then  $	\max\limits_{J \in  \mathcal{I}^{r_{\tilde{\omega}}}}  \big\{\sum_{i \in J} k_{i}  \big \}=r_{\tilde{\omega}} s$ and thus, $ \big\lfloor \frac{d_{(P,w,\pi)} (\mathcal{C})-m_w}{ M_w} \big\rfloor  \leq n - \frac{ \lceil log_{q}|\mathcal{C}| \rceil}{s}$ becomes the \textit{Singleton bound} for a $(P,w,\pi)$-block code $\mathcal{C}$ in which all the blocks are of \textit{same size $s$}. Further, if  $\mathcal{C}$ is a linear code of dimension $k$ then it becomes $ \big\lfloor \frac{d_{(P,w,\pi)} (\mathcal{C})-m_w}{ M_w} \big\rfloor  \leq   n - \frac{k}{s} $. Moreover, when $k_i=1$ $\forall$ $i \in [n] $, we obtain the Singleton bound for $(P,w)$-code:
\begin{corollary}[Singleton Bound for $(P,w)$-code \cite{aks}]
	If $\mathcal{C}$ is a $(P,w)$-code of length $n$ over $\mathbb{F}_q$ with minimum distance $d_{(P,w)}(\mathcal{C})$, then $  \big\lfloor \frac{d_{(P,w)} (\mathcal{C})-m_w}{ M_w} \big\rfloor   \leq n - \lceil log_{q}|\mathcal{C}| \rceil $.
\end{corollary}
\section{Weight Distribution of $(P,w,\pi)$-space}
\label{section:4}
For a weight $w$ on $ \mathbb{F}_{q} $,  the sets  $D_r = \{\alpha \in \mathbb{F}_q : w(\alpha) = r\}$, for $0 \leq r \leq M_w $ partition $ \mathbb{F}_{q} $ according to the distribution of weights $w$ of elements in $ \mathbb{F}_{q} $. Recall that,   $ \tilde{w}^{k} $ is a weight on   $ \mathbb{F}_{q}^{k} $ defined as  $ \tilde{w}^{k}{(u)}=\max\{w(u_{i}) : 1 \leq i \leq k\}$ for $u=(u_{1},u_{2},~\ldots,u_{k}) \in \mathbb{F}_{q}^{k} $.  For each $0 \leq r \leq M_w $, we shall define:
\begin{align*}
	D_{r}^{k}= \{u=(u_{1},u_{2},\ldots,u_{{k}}) \in \mathbb{F}_{q}^{k}  :  \tilde{w}^{k}(u)=r  \}
\end{align*}
to be a subset of $ \mathbb{F}_{q}^{k} $ so that $D_{r}^{k}$ $(0 \leq r \leq M_w )$ partitons $ \mathbb{F}_{q}^{k} $ according to the
distribution of weights $\tilde{w}^{k}$ of elements in 
$ \mathbb{F}_{q}^{k} $. 
Thus, we have the following:
\begin{proposition} \label{D_r^j}
	For a given positive integer $ k $, $|D_r^{k}|=( \sum\limits_{i=0}^{r} |D_{i}|)^{k}-( \sum\limits_{i=0}^{r-1} |D_{i}|)^{k}$ for $0 \leq r \leq M_w $. In particular, 
	$|D_{M_w}^{k}|=q^{k}-(q- |D_{M_w} |)^{k}$.
\end{proposition}
\begin{proof}
	Let $u=(u_{1},u_{2},\ldots,u_{{k}}) \in D_{r}^{k} $ so that $ \tilde{w}^{k}(u)=r$. This means, at least one component of $u$  has  its maximum weight as  $r$. The number of  such $u   \in \mathbb{F}_{q}^{k} $ with $l$ components having  $r$ as their $w$-weight  is $	{k \choose l}|D_{r}|^{l}(1+|D_{1}|+ |D_{2}|+ \ldots + |D_{r-1} |)^{k-l} $. Thus,
	\begin{align*}
		|D_r^{k}|&= \sum\limits_{l=1}^{k} {k \choose l} |D_{r}|^{l} (1+|D_{1}|+ |D_{2}|+ \ldots + |D_{r-1} |)^{k-l}  \\
		&=( 1+ |D_{1}|+ |D_{2}|+ \ldots +  |D_{r}| )^{k}-  ( 1+|D_{1}|+ |D_{2}| + \ldots + |D_{r-1} |)^{k} .
	\end{align*}
	Hence,  $	|D_r^{k}| = ( \sum\limits_{i=0}^{r} |D_{i}|)^{k}-( \sum\limits_{i=0}^{r-1} |D_{i}|)^{k}	$
	where  $D_r = \{\alpha \in \mathbb{F}_q : w(\alpha) = r\}$.
\end{proof}
The above Proposition \ref{D_r^j} gives the weight distribution of $\mathbb{F}_{q}^{k}$ by considering $ \tilde{w}^{k}$ as weight.
\par Next, we shall proceed on to determine the $(P,w,\pi)$-weight distribution of $\mathbb{F}_{q}^{N}$. 
For each $1 \leq r \leq nM_w$, let $A_r = \{x= x_{1} \oplus x_{2} \oplus \ldots \oplus x_{n} \in  \mathbb{F}_{q}^{N}  : w_{(P,w,\pi)}(x)=r  \}$ and $A_0 = \{ \bar{0}\}$. It is clear that if 
$x \in A_r$ then $I_{x}^{P,\pi} = \langle supp_{\pi} (x) \rangle \in \mathcal{I}_j^{i}$ for some $j \in \{1,2,\ldots,n\}$ such that $i \geq j$.
It is also easy to observe the following: 
\begin{itemize}
	\item If $w_{(P,w,\pi)}(x) \leq M_w$  then $I_{x}^{P,\pi} \in \mathcal{I}_j^{j}$ for some $j  \in \{1,2,\ldots,n\}$.
	\item If  $w_{(P,w,\pi)}(x) \geq M_w $  then  $I_{x}^{P,\pi}  \in \mathcal{I}_j^{i}$ for some $j  \in \{1,2,\ldots,n\}$ such that $i \geq j$.
\end{itemize}
Now, for an $x \in A_r$ with $x= x_1 \oplus x_2 \oplus \cdots \oplus x_n $ and  $x_i = (x_{i_1},x_{i_2},\ldots,x_{i_{k_i}}) \in \mathbb{F}_{q}^{k_i} $, we have 
\begin{align*}
	w_{(P,w,\pi)}(x) &= \sum\limits_{i \in M_{x}^{P,\pi}}  \tilde{w}^{k_i}{(x_i)} +   \sum\limits_{i \in {I_x}^{P,\pi} \setminus M_{x}^{P,\pi}} M_w \\
	& =  \tilde{w}^{k_{i_1}}(x_{i_1})+  \tilde{w}^{k_{i_2}}(x_{i_2})  + \ldots+  \tilde{w}^{k_{i_j}}(x_{i_j})  + (i-j)M_w
\end{align*}
where $M_x^{P,\pi}= \{i_1, i_2, \ldots, i_j\}$. 
Then, 
\begin{align*}
	\tilde{w}^{k_{i_1}}(x_{i_1}) + \tilde{w}^{k_{i_2}}(x_{i_2})+\ldots + \tilde{w}^{k_{i_j}}(x_{i_j}) =r-(i-j)M_w.
\end{align*}
As the  weight is varying on the block positions with respect to maximal elements but fixed on block positions with respect to non-maximal elements, we need to choose $j$ number of blocks $x_{i_s} \in  \mathbb{F}_{q}^{k_{i_s}}  $ such that  $\sum\limits_{s=1}^{j} \tilde{w}^{k_{i_s}} (x_{i_s}) =r-(i-j)M_w $.  Thus, to get the  number of  vectors in $A_r$, we first  define the  partition of the number  $r-(i-j)M_w $ in order to have numbers like  $\tilde{w}^{k_{i_1}} (x_{i_1})$, $\tilde{w}^{k_{i_2}} (x_{i_2}), \ldots$, $\tilde{w}^{k_{i_s}} (x_{i_s})$  as its $j$ parts.

\begin{definition}[{Partition of $r$}]
	For any positive integer $r$, by a partition of $r$ we mean a finite non-increasing sequence of positive integers $b_1,b_2,\ldots,b_j $ such that $b_1+b_2+\dots+ b_j =r $ and is denoted by $(b_1,b_2,\ldots,b_j)$. $b_1,b_2,\ldots,b_j $ are called as parts of the partition  and they need not be distinct.
	As the number of parts in a partition is restricted to at most $n$,  such a partition is termed as an $n$-part partition. Let $PRT[r]$ denote the set of all  $n$-part partitions of $r$ such that each part in $r$ does not exceed $M_w$ i.e. 
	\begin{align*}
		PRT[r]=	\{(b_1,b_2,\ldots,b_j) :\sum\limits_{s=1}^{j}b_s=  r; ~ 1 \leq b_s \leq M_{w}  \text{ and }  1 \leq j \leq n\}.
	\end{align*} 
\end{definition}
Since we need to partition $r-(i-j)M_w $, and $ \tilde{w}^{k_{i_1}}(x_{i_1}) + \tilde{w}^{k_{i_2}}(x_{i_2})+\ldots + \tilde{w}^{k_{i_j}}(x_{i_j}) =r-(i-j)M_w$,  the number of parts  is restricted to at most $n-(i-j)$. Hence, to determine the parts of $r-(i-j)M_w$, we define:  $PRT_{i-j}[r]=$
\begin{align*}
	\{(b_1,b_2,\ldots,b_t) : \sum\limits_{s=1}^{t}b_s=  r- (i-j)M_w;  1 \leq b_s \leq M_{w}   \text{ and }   1 \leq t \leq n-i+j\}
\end{align*}
where $i-j$ is the number of non-maximal elements in a given ideal $I \in \mathcal{I}_{j}^{i}$. Note that $PRT[r]=PRT_0[r]$.
\par 
Let  $b= (b_1,b_2,\ldots,b_j) \in PRT_{i-j}[r]$ and  $I \in \mathcal{I}_{j}^{i}$. As the  weight is varying on the $j$ maximal block positions but fixed on the $i-j$ non-maximal block positions, we need to choose $j$ number of $x_{i_s} \in  \mathbb{F}_{q}^{k_{i_s}}  $ such that  $\sum\limits_{s=1}^{j} \tilde{w}^{k_{i_s}} (x_{i_s}) =r-(i-j)M_w $. 
Parts $b_{i}$'s of $b$ do help in finding suitable  $x_{i} \in  \mathbb{F}_{q}^{k_{i}}  $  that  have weight   $ \tilde{w}^{k_i} (x_{i}) = b_{i}$ so that all these $x_{i} $'s will occur only in the positions of maximal elements of  $I$. Since weights in the positions of non-maximal elements are fixed as $M_w$,  any  $x_{i} \in  \mathbb{F}_{q}^{k_{i}}  $ can occcur in those non-maximal positions. But in the remaining $n-|I|$ positions $l$,  $\bar{0} \in  \mathbb{F}_q^{k_l}$ will  occur. Thus, for a $b= (b_1,b_2,\ldots,b_j) \in PRT_{i-j}[r]$ and  $I \in \mathcal{I}_{j}^{i}$, the set 
\begin{equation*}
	\begin{split}
		A = 	\bigg\{    x=  x_{1}  \oplus x_{2}  \oplus \cdots  \oplus x_{n}    \in \mathbb{F}_{q}^{N} : 
		x_s  =  \left\{ 
		\begin{array}{ll}
			\bar{0}  \in \mathbb{F}_q^{k_s}, & \text{for} ~s \in [n] \setminus I \\
			x_s \in	D_{b_{s}}^{k_s}, &   \text{for}  \ s \in Max(I) \\
			x_s	\in \mathbb{F}_q^{k_s}, &  \text{for}  \ s \in I  \setminus Max(I)
		\end{array}	\right\}, \\ 1 \leq s \leq n \bigg\}
	\end{split}
\end{equation*}
gives the set of those $N$-tuples $x= x_{1} \oplus x_{2} \oplus \ldots \oplus x_{n} \in \mathbb{F}_{q}^{N}$ each of whose $(P,w,\pi)$-weight is $r$, $ \langle supp_{\pi} (x) \rangle = I$ and $ \tilde{w}^{k_s} (x_{s}) = b_{s}$ for $s \in Max(I)$. Such a set $A$ of vectors of $(P,w,\pi)$-weight $r$ can be obtained for each arrangement of a $b \in PRT_{i-j} [r]$ for a given $I \in \mathcal{I}_j^{i} $. 
\begin{example} \label{ex1}
	Let $\mathbb{Z}_{7}^{13}=\mathbb{Z}_{7}^2 \oplus \mathbb{Z}_{7}^3 \oplus \mathbb{Z}_{7}^4 \oplus \mathbb{Z}_{7}^2 \oplus \mathbb{Z}_{7}^2$. Consider $w$ as the Lee weight on $\mathbb{Z}_{7}$. Let $\preceq $ be a partial order relation  on the set $[5]=\{1,2,3,4,5\}$ such that $1 \preceq 2$. Here $I = \{1,2,4\}$  is the ideal of cardinality $3$ with $2$ maximal elements. $Max(I)= \{2,4\}$, $I \setminus Max(I)= \{1\}$ and $[n] - I =\{3,5\}$. As $M_w = 3$, 
	$	PRT_1[8] = \{(3,2),(3,1,1), (2,2,1),(2,1,1,1)\}$. Now, for $b=(3,2) \in  PRT_1[8] $ and  $I = \{1,2,4\}$, at the positions of maximal elements  $Max(I)= \{2,4\}$, $x_2 \in \mathbb{Z}_{7}^3$ such that  $ \tilde{w}^{k_2} (x_{2}) = 3$, and $x_4 \in \mathbb{Z}_{7}^2$ such that $ \tilde{w}^{k_4} (x_{4}) = 2$ will occur. At non-maximal positions of $I$ (here it is $ \{1\}$),  any $x_1 \in \mathbb{Z}_{7}^2$ will occur. In the remaining positions $[n] - I =\{3,5\}$, $x_3 =\bar{0} \in \mathbb{Z}_{7}^4$ and  $x_3 =\bar{0} \in \mathbb{Z}_{7}^2$ will occur. Thus, when  $I = \{1,2,4\}$ and  $b=(3,2) \in  PRT_1[8] $, 
	\begin{equation*}
		\begin{split}
			A =	\{x=x_1 \oplus x_2 \oplus \bar{0} \oplus x_4 \oplus \bar{0}: x_1 \in \mathbb{Z}_{7}^2, ~ x_2 \in \mathbb{Z}_{7}^3  \text{ with }  \tilde{w}^{k_2} (x_{2}) = 3, \\ x_4 \in \mathbb{Z}_{7}^2 \text{ with }   \tilde{w}^{k_4} (x_{4}) = 2 \}
		\end{split}
	\end{equation*}
	gives the set of vectors that have  $(P,w,\pi)$-weight  $3+3+2=8$. 
	And,  if we take $(2,3)$ which is an arrangement of parts of $b=(3,2)$ then 
	\begin{align*}
			B=	\{x=x_1 \oplus x_2 \oplus \bar{0} \oplus x_4 \oplus \bar{0} : x_1 \in \mathbb{Z}_{7}^2, ~ x_2 \in \mathbb{Z}_{7}^3 \text{ with }  \tilde{w}^{k_2} (x_{2}) = 2, \\ ~ x_4 \in \mathbb{Z}_{7}^2 ~\text{ with }  \tilde{w}^{k_4} (x_{4}) = 3 \}
	\end{align*}
	also gives the set of vectors that have $(P,w,\pi)$-weight $3+2+3=8$ when  $I = \{1,2,4\}$.
\end{example}
\begin{definition}[{Arrangement of $b$}]
	Given an $n$-part partition $ b= (b_1,b_2,\ldots,b_j) \in PRT_{i-j}[r] $ with $j$ parts, an arrangement of $b$ is a $j$-tuple   $ (b_{t_1},b_{t_2},\ldots,b_{t_j}) $ such that $t_i \in \{1,2,\ldots,j\}$ are distinct. Thus, $ARG[(b_1,b_2,\ldots,b_j)] =$
	\begin{align*}
			\{ (b_{t_1},  b_{t_2},  \ldots , b_{t_j}): t_1, t_2, \ldots,t_j \in  \{1,2,\ldots,j\},~ t_i \neq t_k~ \text{for} ~ i \neq k \}
	\end{align*}
	denote the set of all such arrangements of a given $b= (b_1,b_2,\ldots,b_j) \in PRT_{i-j}[r] $.  \par Partition of a weight $r$ (or  $r-(i-j)M_w$, as the case may be) and arrangement of their parts play a vital role in determining the weight distribution of a $(P,w,\pi)$-space.
\end{definition}
\begin{example}
	Continuing from the Example \ref{example1} with $r=11$, $M_w =3$ and $n=5$, we can see that the arrangements of  $b= (3,3,3,2) \in PRT_0[11] $ are given by
	\begin{align*}
		ARG[(3,3,3,2)] = \{ (3,3,3,2),(3,3,2,3), (3,2,3,3), (2,3,3,3)\} 
	\end{align*}	
	and the arrangements of  $b= (3,3,3,1,1) \in PRT_{0}[11] $ are given as
	\begin{align*}
			ARG[(3,3,3,1,1)] = \{ (3,3,3,1,1),(3,3,1,1,3), 
			(3,1,1,3,3),(1,1,3,3,3), \\ (1,3,1,3,3),  (1,3,3,1,3),(1,3,3,3,1), 
			(3,1,3,1,3),  (3,1,3,3,1), (3,3,1,3,1)\}.
	\end{align*}
\end{example}
\par It is to be the noted that, if $t_1, t_2 , \ldots, t_l$ denote the  $l$ distinct elements in the parts  $b_1,b_2,\ldots,  b_t$ with multiplicity   $r_1,r_2,\ldots,r_l$ respectively so that   $b_1+b_2+\dots+ b_t = \sum\limits_{s=1}^{l} r_s  t_s= r - (i-j) M_w  $, then  $| ARG[b]|= \frac{t!}{r_1 ! \ r_2 ! \ \cdots r_l !}$ for each $b=(b_1,b_2,\ldots,b_t)  \in PRT_{i-j}[r]$.
\par Now, for each $b=(b_1,b_2,\ldots,b_t) \in 	PRT_{i-j}[r]$ and for each $I  \in \mathcal{I}_j^{i}$, we  define a set  $T_b [I] $ which gives all  vectors $x =  x_{1} \oplus x_{2} \oplus \ldots \oplus x_{n} \in \mathbb{F}_q^{N}$ with $(P,w,\pi)$-weight $ r=b_1 + b_2 +  \ldots + b_j +(i-j)M_w  $ such that  $\langle supp_{\pi} (x) \rangle=I $. Thus, 
\begin{equation*}
	\begin{split}
		T_b[I] =	 &\bigcup\limits_{ (b_{m_{i_1}},  b_{m_{i_2}},  \ldots , b_{m_{i_j}})    \in ARG[b] }  \bigg\{   x= x_{1}   \oplus \cdots  \oplus x_{n}    \in \mathbb{F}_{q}^{N} :  
		x_s  = \\& \hspace*{3cm}  \left\{ 
		\begin{array}{ll}
			\bar{0}  \in \mathbb{F}_q^{k_s}, & \text{for} ~s \notin I \\
			x_s \in	D_{b_{m_s}}^{k_s}, &   \text{for}  \ s \in Max(I) \\
			x_s	\in \mathbb{F}_q^{k_s}, &  \text{for}  \ s \in I  \setminus Max(I)
		\end{array}	\right\},    1 \leq s \leq n \bigg\}.
	\end{split}
\end{equation*}
\begin{proposition}
	For $1 \leq r \leq n M_{w}$, we have   
	$
	A_r =  \bigcup\limits_{i =1}^{n}  \bigcup\limits_{j=1 }^{i}  \bigcup\limits_{I \in \mathcal{I}_j^{i} } \bigcup\limits_{b =(b_1,b_2,\ldots,b_t) \in 	PRT_{i-j}[r] } T_b [I].$
\end{proposition}
When $r \leq M_w$,  for any $x \in A_r$, all elements of $I_x^{P,\pi}$ are  maximal (that is, $I_x^{P,\pi}= M_x^{P,\pi}$), so that   $w_{(P,w,\pi)}(x)= \sum\limits_{i \in M_{x}^{P,\pi}}  \tilde{w}^{k_i}{(x_i)}  =  \tilde{w}^{k_{i_1}}(x_{i_1})+  \tilde{w}^{k_{i_2}}(x_{i_2})  + \cdots+  \tilde{w}^{k_{i_j}}(x_{i_j})$ where $ M_x^{P,\pi} = \{i_1,i_2,\ldots,i_j\} $. 
\begin{theorem}\label{rlessthanMw}
	For any $1 \leq r \leq M_{w} $,  the number of $N$-tuples  $x \in \mathbb{F}_{q}^N$ having $w_{(P,w,\pi)}(x)= r$  is  $A_r =$
	\begin{align*}
		\sum\limits_{j=1}^{n} \sum\limits_{I \in  \mathcal{I}_j^{j}} \sum\limits_{(b_1,b_2,\ldots,b_j) \in PRT_0[r]} \sum\limits_{(b_{m_1},  b_{m_2},  \ldots , b_{m_j}) \in ARG[(b_1,b_2,\ldots,b_j)] }  |D_{b_{m_{i_1}}}^{k_{i_1}}| |D_{b_{m_{i_2}}}^{k_{i_2}}| \cdots |D_{b_{m_{i_j}}}^{k_{i_j}}|
	\end{align*}
	wherein, for an $I \in \mathcal{I}_j^{j}$, $Max(I) = \{i_1,i_2,\ldots,i_j\}$ and $(b_{1},b_{2}, \dots, b_{j}) \in PRT_{0}[r]$.
\end{theorem}
\begin{proof} 
	Let $x \in {A}_r $ and $I = \langle supp_{\pi} (x) \rangle $. Now $I  \in \mathcal{I}_j^{j}$ for some $j \in \{1,2,\ldots,n\}$ as $r \leq M_w$. Let $ I = \{i_1,i_2,\ldots,i_j\}$, $b=(b_1,b_2,\ldots,b_j)  \in PRT_{0}[r]$.  Thus,
	\begin{align*}
		T_b [I] =	&	 \bigcup\limits_{ (b_{m_{i_1}},  b_{m_{i_2}},  \ldots , b_{m_{i_j}})    \in ARG[b] }  		
		\bigg\{ x= x_{1} \oplus x_{2} \oplus \cdots \oplus x_{n} \in \mathbb{F}_{q}^{N} :  
		x_s  = \\& \hspace*{5cm}\left\{ 
		\begin{array}{ll}
			0  \in \mathbb{F}_q^{k_s}, & \text{for} ~s \notin I \\
			x_s \in 	D_{b_{m_s}}^{k_s}, &   \text{for}  \ s \in I
		\end{array}	\right\}, 1 \leq s \leq n \bigg\}
	\end{align*}
	gives all the vectors $x \in \mathbb{F}_q^{N}$ having $w_{(P,w,\pi)}(x)= r$ such that  $ \langle supp_{\pi} (x) \rangle =I $ and $ \sum\limits_{s \in Max(I)}  \tilde{w}^{k_s}(x_s) = b_1 + b_2 +  \cdots + b_j$. So,   
	\begin{align*}
		|T_b [I] | &= \sum\limits_{(b_{m_{i_1}},  b_{m_{i_2}},  \ldots , b_{m_{i_j}}) \in ARG[(b_{1},b_{2}, \dots, b_{j})] } |D_{b_{m_{i_1}}}^{k_{i_1}}|  |D_{b_{m_{i_2}}}^{k_{i_2}}| \cdots  |D_{b_{m_{i_j}}}^{k_{i_j}}| .
	\end{align*}
	Hence, the number of $x \in \mathbb{F}_q^{N}$ having ${(P,w,\pi)}$-weight as $r$ is
	$	|A_{r}| 
	= \sum\limits_{j=1}^{n} \sum\limits_{I \in  \mathcal{I}_j^{j}} \sum\limits_{(b_1,b_2,\ldots,b_j) \in PRT_0[r]} |T_b [I] | .$
\end{proof} 
For the case when  $r \geq M_w$, we will have $ r=tM_w+a$ for some non-negative integer $a \leq M_w$ and  for any $x \in A_r$, $I_x^{P, \pi} $ will consist of at most $t$ non-maximal elements.
\begin{theorem}\label{r=Mw}
	For any $ r =t M_{w} + a$ with $1 \leq t \leq n-1$ and $0 < a \leq M_w$, the number of $N$-tuples  $ x \in \mathbb{F}_{q}^N$ having $w_{(P,w,\pi)}(x)= r$  is 
	\begin{align*}
		|A_{tM_w + a}| =	\sum\limits_{i=0}^{t} \sum\limits_{j=1}^{n}  \sum\limits_{I \in  \mathcal{I}_j^{j+i}}
		\sum\limits_{(b_1,b_2,\ldots,b_j) \in PRT_{i}[r]}    \sum\limits_{(b_{m_{i_1}},b_{m_{i_2}},  \ldots , b_{m_{i_j}}) \in ARG[(b_1,b_2,\ldots,b_j)] }  |D_{b_{m_{i_1}}}^{k_{i_1}}|  \\  |D_{b_{m_{i_2}}}^{k_{i_2}}|  \cdots |D_{b_{m_{i_j}}}^{k_{i_j}}|   q^{k_{l_1} + k_{l_2} + \cdots + k_{l_i}} 
	\end{align*} 
	wherein, for an $I \in {\mathcal{I}_j^{j+i} }$,  $Max(I)=\{i_1,i_2,\ldots,i_j\}$ and  $I \setminus Max(I) = \{l_1, l_2, \dots, l_{i} \}$.
\end{theorem}
\begin{proof}
	Let $x \in {A}_r $ and $I = \langle supp_{\pi} (x) \rangle $. As  $w_{(P,w,\pi)}(x) =r \geq M_w  $, $I  \in \mathcal{I}_j^{j+i}$ for some $j  \in \{1,2,\ldots,n\}$ and $ i \geq  0$.  
	Let $ Max(I) = \{i_1,i_2,\ldots,i_j\} $, $ I  \setminus Max(I) = \{l_1, l_2, \dots, l_{i} \}$ and $b=(b_1,b_2,\ldots,b_j)  \in PRT_{i}[r]$.  Thus, 
	\begin{equation*}
		{{\begin{split}
					T_b[I] =	 	\bigcup\limits_{ (b_{m_{i_1}},  b_{m_{i_2}}, \ldots , b_{m_{i_j}})    \in ARG[b] } 
					\bigg\{ &
					x= x_{1} \oplus  \cdots \oplus x_{n} \in \mathbb{F}_{q}^{N} :  	  
					x_s  = \\ & \hspace*{-1cm}\left\{ 
					\begin{array}{ll}
						\bar{0}  \in \mathbb{F}_q^{k_s}, & \text{for} ~s \notin I \\
						x_s \in	D_{b_{m_s}}^{k_s}, &  \text{for}  \ s \in Max(I) \\
						x_s	\in \mathbb{F}_q^{k_s}, &  \text{for}  \ s \in I  \setminus Max(I)
					\end{array}	\right\},  1 \leq s \leq n  \bigg\}  
		\end{split}}}
	\end{equation*}
	\noindent gives all the vectors $x \in \mathbb{F}_q^{N}$ having $w_{(P,w,\pi)}(x)= r$ such that  $\langle supp_{\pi} (x) \rangle  = I$ and $ \sum\limits_{i \in Max(I)}  \tilde{w}^{k_i}(x_i) = b_1 + b_2 +  \cdots + b_j$. So,   
	\begin{align*}
		|T_b [I]| =& \sum\limits_{(b_{m_{i_1}},  b_{m_{i_2}},  \ldots , b_{m_{i_j}})) \in ARG[b] } |D_{b_{m_{i_1}}}^{k_{i_1}}|  |D_{b_{m_{i_2}}}^{k_{i_2}}| \cdots |D_{b_{m_{i_j}}}^{k_{i_j}}| q^{k_{l_1} + k_{l_2} + \cdots + k_{l_i}}. 
	\end{align*}	 
	As   $0 \leq |I \setminus Max(I)| = i \leq t $, so 
	\begin{equation*}
		|A_{r}| 
		=  \sum\limits_{i=0}^{t} \sum\limits_{j=1}^{n} \sum\limits_{I \in  \mathcal{I}_j^{j+i}} \sum\limits_{(b_1,b_2,\ldots,b_j) \in PRT_i[r]} |T_b [I]|.
	\end{equation*}
\end{proof}
Based on the results arrived at, $|A_r|$ can be stated  for certain particular values of $r$: \\
When $r < M_w$, we have
\begin{align*}
	|A_1| &= \sum\limits_{I \in  \mathcal{I}_1^{1}}  |D_{1}^{k_{i_1}} |, \\
	|A_2| &= \sum\limits_{I \in  \mathcal{I}_1^{1}} |D_{2}^{k_{i_1}} | + \sum\limits_{I \in  \mathcal{I}_2^{2}}  |D_{1}^{k_{i_1}} | |D_{1}^{k_{i_2}} |\\ &\text{ and so on}.
\end{align*}	
When $r = M_w$, we have   
\begin{equation*}
	\begin{split}
		 |A_{M_w}| =  \sum\limits_{j=1}^{n} \sum\limits_{I \in  \mathcal{I}_j^{j}} 
		\sum\limits_{(b_1,b_2,\ldots,b_j) \in PRT_0{[M_w]}} 
		\sum\limits_{(b_{m_{i_1}}, b_{m_{i_2}},   \ldots , b_{m_{i_j}}) \in ARG[(b_1,b_2,\ldots,b_j)] }
		|D_{b_{m_{i_1}}}^{k_{i_1}}| \\ |D_{b_{m_{i_2}}}^{k_{i_2}}|  \cdots |D_{b_{m_{i_j}}}^{k_{i_j}}|  
	\end{split}
\end{equation*}
When  $r=M_{w}+a$ for  $1\leq a \leq M_{w} -1$, we have 
\begin{align*}
	|A_{M_{w} + a}|= 	\sum\limits_{j=1}^{n} \sum\limits_{I \in  \mathcal{I}_j^{j}} \sum\limits_{(b_1,b_2,\ldots,b_j) \in PRT_0[r]} 
	\sum\limits_{(b_{m_{i_1}}, b_{m_{i_2}},    \ldots , b_{m_{i_j}}) \in ARG[(b_1,b_2,\ldots,b_j)] }  |D_{b_{m_{i_1}}}^{k_{i_1}}| \\
	|D_{b_{m_{i_2}}}^{k_{i_2}}|   \cdots 
	|D_{b_{m_{i_j}}}^{k_{i_j}}|      +
	\sum\limits_{j=1}^{n} \sum\limits_{I \in  \mathcal{I}_j^{j+1}}   \sum\limits_{(b'_1,b'_2,\ldots,b'_j) \in PRT_{1}[r]}  \sum\limits_{(b'_{m_{i_1}},b'_{m_{i_2}},   \ldots , b'_{m_{i_j}}) \in ARG[(b'_1,b'_2,\ldots,b'_j)] } \\
	|D_{b'_{m_{i_1}}}^{k_{i_1}}|  |D_{b'_{m_{i_2}}}^{k_{i_2}}|   \cdots |D_{b'_{m_{i_j}}}^{k_{i_j}}|    q^{k_{l_1}}.		
\end{align*}
\par Now we shall illustrate the obtained results by choosing the field as $\mathbb{Z}_{7}$ by considering the weight $w$ on it as the Lee weight.  Note that, if we let $L_j = \{ i \in [n] : | \langle i \rangle |= j\} $,	 $1 \leq j \leq n$, then  $|\mathcal{I}_j^j| = {|L _1| \choose j}$. If $m_0$ denote the number of minimal elements in the given poset $P$, then  $|L_1|=m_0$. 
\begin{example} \label{exa1}
	Let $\preceq $ be a partial order relation  on  $[5]=\{1,2,3,4,5\}$ such that $1 \preceq 2$. Consider $w$ as the Lee weight on $\mathbb{Z}_{7}$ and let $\mathbb{Z}_{7}^{13}=\mathbb{Z}_{7}^2 \oplus \mathbb{Z}_{7}^3 \oplus \mathbb{Z}_{7}^4 \oplus \mathbb{Z}_{7}^2 \oplus \mathbb{Z}_{7}^2$.  Here $I_1 = \{1\}$, $I_2 = \{3\}$, $I_3 = \{4\}$, $I_4 = \{5\}$  are the only ideals with cardinality one with respect to 
	$ \preceq$ and hence $\mathcal{I}_1^1 =\{I_1, I_2, I_3,I_4\}$. $I_5 = \{1,2\}$ is the ideal with cardinality two with one maximal element and $\mathcal{I}_1^{2} =\{I_5\}$.  
	$I_6 = \{1,3\}$, $I_7 = \{1,4\}$, $I_8 = \{1,5\}$, $I_9 = \{3,4\}$, $I_{10} = \{3,5\}$, $I_{11} = \{4,5\}$ are the ideals with cardinality $2$ with $2$ maximal elements and $\mathcal{I}_2 ^2 =\{I_6, I_7, I_8,I_{9},I_{10},I_{11}\}$. 
	$I_{12} = \{1,2,3\}$, $I_{13} = \{1,2,4\}$, $I_{14} = \{1,2,5\}$ are the ideals with cardinality $3$ with $2$ maximal elements and $\mathcal{I}_2^{3} =\{I_{12}, I_{13}, I_{14}\}$.
	$I_{15}= \{1,3,4\}$,  $I_{16}= \{1,3,5\}$, $I_{17}= \{1,4,5\}$, $I_{18}= \{3,4,5\}$ are the ideals with cardinality $3$ with $3$ maximal elements and $\mathcal{I}_3^3 = \{I_{15},I_{16}, I_{17}, I_{18} \}$. 
	$I_{19} = \{1,2,3,4\}$, $I_{20} = \{1,2,3,5\}$, $I_{21} = \{1,2,4,5\}$ are the ideals with cardinality $4$ with $3$ maximal elements and  $\mathcal{I}_3 ^4 =\{I_{19}, I_{20}, I_{21} \}$.  
	$I_{22} = \{1,3,4,5\}$ is the only ideal with cardinality $4$ with $4$ maximal elements and thus $\mathcal{I}_4 ^4 =\{I_{22} \}$.  $I_{23} = \{1,2,3,4,5\}$ is the only ideal with cardinality $5$ with $4$ maximal elements and hence $\mathcal{I}_4 ^5 =\{I_{23} \}$. And, $\mathcal{I}_1^3 = \mathcal{I}_1^4 = \mathcal{I}_1^5 = \mathcal{I}_2^4 = \mathcal{I}_2^5 =  \mathcal{I}_3^5 =  \mathcal{I}_5^5 = \{ \}$. 
	\begin{enumerate}[label=(\roman*)]
		\item   \sloppy{In $\mathbb{Z}_7$, $M_w=3$ and $|D_{1}| = |D_{2}| =|D_{3}| =2 $. 
			From Proposition \ref{D_r^j}, in $\mathbb{Z}_7^{13}$,} \\ 
		$|D_{1}^{k}| = (1+|D_{1}| )^{k} -1$, \\
		$|D_{2}^{k} | = (1+|D_{1}|+|D_{2}| )^{k} -(1+|D_{1}|)^k$, \\
		$|D_{3}^{k} | = (1+|D_{1}|+|D_{2}| +|D_{3}|)^{k} -(1+|D_{1}|+|D_{2}|)^k$, 
		where $k =2,3$, $4$.
		\item 	    $PRT[3]=\{(1 ,1 ,1),  ( 1,2), (3)\}$ and   $ARG[( 1,2)] = \{ ( 1,2), ( 2,1) \}$. 
		Thus,	number of vectors $x \in  \mathbb{Z}_{7}^{13}$ having weight $3$ is 
		\begin{equation*}
			\begin{split}
				|A_3| 	=&	 \sum\limits_{j=1}^{3} \sum\limits_{I \in  \mathcal{I}_j^{j}} \sum\limits_{(b_1,b_2,\ldots,b_j) \in PRT_{0}[r]}  \sum\limits_{(b_{m_{i_1}}, b_{m_{i_2}},   \ldots , b_{m_{i_j}}) \in ARG[(b_1,\ldots,b_j)] }  |D_{b_{m_{i_1}}}^{k_{i_1}}| 
				\cdots |D_{b_{m_{i_j}}}^{k_{i_j}}|  \\ 
				=&   \sum\limits_{I \in  \mathcal{I}_1^{1}} |D_{3}^{k_{i_1}}| + \sum\limits_{I \in  \mathcal{I}_2^{2}} (|D_{1}^{k_{i_1}}|  |D_{2}^{k_{i_2}}| + |D_{2}^{k_{i_1}}|  |D_{1}^{k_{i_2}}|)    + \sum\limits_{I \in  \mathcal{I}_{3}^{3}}  |D_{1}^{k_{i_1}}|  |D_{1}^{k_{i_2}}|  |D_{1}^{k_{i_3}}|  					\\
				=  &	|D_{3}^{2} | + |D_{3}^{4} | + |D_{3}^{2} |  +|D_{3}^{2} | 
				+  |D_{1}^{2} | |D_{2}^{4} | +  |D_{2}^{2} | |D_{1}^{4} | 
				+ |D_{1}^{2} | |D_{2}^{2} | + |D_{2}^{2} | |D_{1}^{2} | + \\& 
				|D_{1}^{2} |  |D_{2}^{2} |  +    |D_{2}^{2} | |D_{1}^{2} | +  
				|D_{1}^{4} | |D_{2}^{2} |  +   |D_{2}^{4} | |D_{1}^{2} | +   |D_{1}^{4} | |D_{2}^{2} |  + |D_{2}^{4} | |D_{1}^{2} | + \\& |D_{1}^{2} |  |D_{2}^{2} | 		  
				+ |D_{2}^{2} | |D_{1}^{2} |  +     |D_{1}^{2} |   |D_{1}^{4} |  |D_{1}^{2} |  +   |D_{1}^{2} |   |D_{1}^{4} |   |D_{1}^{2} | 	+ 	 |D_{1}^{2} |   |D_{1}^{2} |  |D_{1}^{2}| 	
				+  \\&|D_{1}^{4} |  |D_{1}^{2} | |D_{1}^{2}|   
				\\
				= & 24 + 1776 + 24 + 24 +  8 \times 544 +  16 \times 80 +  8 \times 16 +  16 \times 8 +   8 \times 16   + \\& 16 \times 8 +  80 \times 16   +
				544 \times 8   +  80 \times 16 + 544 \times 8 +  8 \times 16 + 
				16 \times 8 +  \\&
				8 \times 80  \times 8 +     8 \times 80  \times 8
				+  8 \times 8  \times 8   + 80 \times 8  \times 8  \\
				= & 35,384.								  						
			\end{split}
		\end{equation*}
		\item      Also,  $PRT[14]=\{(3,3,3,3,2)\}$ and $PRT_1[14]=\{(3,3,3,2)\}$.   Consider the arrangements of $ (3,3,3,2)$ as $ARG[(3,3,3,2)] = \{ (3,3,3,2),(3,3,2,3),(3,2,3,3),(2,3,3,3)\}$. Now  $\mathcal{I}_4 ^5 =\{I_{23} \}$ and $\mathcal{I}_5^5 = \{\phi\}$.  Thus,	number of vectors $x \in  \mathbb{Z}_{7}^{13}$ having weight $14$ is  
		\begin{equation*}
			\begin{split}			
				|A_{14}| &= 		\sum\limits_{I \in  \mathcal{I}_{4}^{5}} 			  \sum\limits_{(b_1,b_2,b_3,b_4) \in PRT_{1}[14]}  \sum\limits_{(b_{m_{i_1}},  b_{m_{i_2}},  b_{m_{i_3}} , b_{m_{i_4}}) \in ARG[(3,3,3,2)] }  \\ & \hspace*{6.5cm} |D_{b_{m_{i_1}}}^{k_{i_1}}|    |D_{b_{m_{i_2}}}^{k_{i_2}}|  |D_{b_{m_{i_3}}}^{k_{i_3}}| |D_{b_{m_{i_4}}}^{k_{i_4}}|  q^{k_{l_1}} \\
				&=(|D_{3}^{3}|  |D_{3}^{4}| |D_{3}^{2}|  |D_{2}^{2}|   +  |D_{3}^{3}| |D_{3}^{4}|  |D_{2}^{2}|  |D_{3}^{2}| +  |D_{3}^{3}|  |D_{2}^{4}|  |D_{3}^{2}|  |D_{3}^{2}|  + \\&
				\hspace*{.7cm} |D_{2}^{3}|  |D_{3}^{4}| |D_{3}^{2}| |D_{3}^{2}|)  7^2\\ 
				&= (218 \times 1776 \times 24 \times 16 + 218 \times 1776 \times  16 \times 24  + 218 \times 544 \times 24 \times 24  + \\& \hspace*{.5cm} 98 \times 1776 \times 24 \times 24 )
				\times 49 \\
				&= 22,829,377,536.
			\end{split}
		\end{equation*}			
	\end{enumerate}
\end{example}
Calculating $ARG[b]$ for each  partition $ b= (b_1,b_2,\ldots,b_j)  \in 	PRT_{i-j}[r] $ is seemingly cumbersome. But, that is not the case  when   $k_i = k$ $\forall$ $i \in [n]$, as the  number of components $x_i \in \mathbb{F}_{q}^{k_i} $ having $\tilde{w}^{k_i}(x_i)=b_i$ will not then depend on the label $i \in [n]$ and thus it is not  necessary to  find  $ARG[b]$. If $ b= (b_1,b_2,\ldots,b_j)  \in 	PRT_{i-j}[r] $, let $t_1, t_2 , \ldots, t_l$ be the  distinct  $l$ elements from  the $j$ parts  $b_1,b_2,\ldots,  b_j$ with multiplicity   $r_1,r_2,\ldots,r_l$ respectively so that   $ \sum_{s=1}^{l} r_s  t_s=b_1+b_2+\dots+ b_j $.
\begin{theorem}\label{Ar without ARG}
	Let $k_i=k$ $\forall$ $i \in [n]$. Then, for any $ r \leq n M_w$, the number of $N$-tuples  $ x \in \mathbb{F}_{q}^N$ having $w_{(P,w,\pi)}(x)= r$  is $|A_{r}| =$
	\begin{align*}
			\sum\limits_{i=1}^{n} \sum\limits_{j=1}^{i} \sum\limits_{I \in \mathcal{I}_j^{i}}  \sum\limits_{(b_1,b_2,\ldots,b_j) \in PRT_{i-j}[r]} q^{k(i-j)}   
		\prod\limits_{s=1}^{l} |D_{t_s}^k | ^{r_s}     {j - (r_1+r_2+\ldots+r_{s-1})\choose r_s} 
	\end{align*}
	where  $r_s$ parts among the $j$ parts in  $b_1, b_2, \dots, b_j  $  are equal to $t_s$ 
	$(1 \leq s \leq l)$.
\end{theorem}
\begin{proof}
	Let $I \in \mathcal{I}_{j}^{i}$ with  $|Max(I)|=j$ and let $ b=(b_1,b_2,\ldots,b_j)  \in PRT _{i-j}[r] $.  Let $t_1, t_2 , \ldots, t_l$ be the distinct $l$  elements from the set of $j$ parts in $b$ such that $r_s$ is the multiplicity of  $t_s$ ($1 \leq s \leq l$) so that $\sum\limits_{s=1}^{l}  r_s  = j $ and $\sum\limits_{s=1}^{l} r_s t_s  = b_1+b_2+\cdots +b_j =r - (i-j) M_w $.
	Since $t_s$  occurs $r_s$ times, the total number of choices for $x_s \in \mathbb{F}_{q}^{k_s} $ such that $\tilde{w}^{k_s} (x_s)=t_s$ in $r_s$ places of $Max(I)$ is $|D_{t_s}^k | ^{r_s} {j \choose r_s}$. Then for each $ (b_1,b_2,\ldots,b_j)  \in PRT _{i-j}[r] $ and $I \in \mathcal{I}_{j}^{i}$, the total number of $N$-tuples $x \in \mathbb{F}_q ^N$ having $w_{(P,w,\pi)}(x)= r $ is $ q^{k(i-j)} \prod\limits_{s=1}^{l}|D_{t_s}^k | ^{r_s} {j - (r_1+r_2+\cdots+r_{s-1})\choose r_s} $. 
	Hence, $|A_{r}| =$
	\begin{align*}
			\sum\limits_{i=1}^{n} \sum\limits_{j=1}^{i} \sum\limits_{I \in \mathcal{I}_j^{i}} \sum\limits_{ (b_1,b_2,\ldots,b_j)  \in PRT _{i-j} [r] }  q^{k(i-j)}  \prod\limits_{s=1}^{l} 
		|D_{t_s}^k | ^{r_s}  {j - (r_1+r_2+\ldots+r_{s-1})\choose r_s} 
	\end{align*} where  $r_0 = 0$ and $|D_{t_s}^k |= ( \sum\limits_{i=0}^{t_s} |D_{i}|)^{k}-( \sum\limits_{i=0}^{t_s-1} |D_{i}|)^k$.
\end{proof}
\begin{corollary}
	If $k_i=k$ $\forall$ $i \in [n]$ then the number of $N$-tuples  $ x \in \mathbb{F}_{q}^N$ having $w_{(P,w,\pi)}(x)= n M_w$  is 
	$
	({ q^{k}-( q -  |D_{M_{w}}|) ^{k} })^{t} q^{k(n-t)}  
	$ 
	where $t$ is the  number of maximal elements in the given poset $P$. 
\end{corollary}
As the $d_{(P,w,\pi)}$-metric is translation invariant, $B_{r}(x) = x + B_{r}(0)$ for an $ x \in \mathbb{F}_{q}^N $ where $ 0 \in \mathbb{F}_{q}^N $.  Moreover,  $| B_{r}(0) | = 1 + \sum\limits_{t=1}^{r} 	|S_t (0)|$ and $ |S_t (0)| =|A_t| $. Thus,
\begin{proposition} \label{wpirball}
	The number of $N$-tuples in a  $(P,w,\pi)$-ball of radius $r$ centered at $ x \in \mathbb{F}_{q}^N $ is  $| B_{r}(x) | = 1 + \sum\limits_{t=1}^{r} 	|A_t|$. 
\end{proposition} 
Note that the results obtained in this section are in some sense generalizations to those obtained by M. M. Skriganov in \cite{skri}. 
In the remaining part of this section, we show how  $|A_r| $ can be arrived at for various spaces viz., $(P,w)$-space, $(P,\pi)$-space, $\pi$-space and $P$-space.
\subsection{$(P,w)$-space}
The  $(P,w,\pi)$-space becomes $(P,w)$-space \cite{wcps} by considering   $k_i =1$ for all $i \in [n]$ (see Remark \ref{becomes}). Here $N=\sum\limits_{i=1}^{n} k_i =n$ . For $u =( u_1 , u_2 , \ldots ,u_n ) \in \mathbb{F}_{q}^{n} $,  $ supp_{\pi}( u)=\{i \in [n] : u_i \neq 0\} = supp(u)$ and $ \tilde{w}^{k_i}{(u_i)} =  w{(u_i)} ~ \forall ~i \in [n]$. Thus,    
\begin{equation*}
	w_{(P,w)}(u)= \sum\limits_{i \in M_{u}^P}  w{(u_i)} + \sum\limits_{i \in {I_u}^{P} \setminus M_{u}^P} M_w.  
\end{equation*}
\par 
Now, $A_r = \{u= ( u_1 , u_2 , \ldots ,u_n ) \in \mathbb{F}_{q}^{n}  : w_{(P,w)}(u)=r  \}$
and by Theorem \ref{rlessthanMw}, for any $1 \leq r \leq M_{w} -1$, the number of $N$-tuples  in $ \mathbb{F}_{q}^n$ having $w_{(P,w)}(x)= r$  is
\begin{align*}
	|A_{r}| 
	&= \sum\limits_{j=1}^{n} \sum\limits_{I \in  \mathcal{I}_j^{j}} \sum\limits_{(b_1,b_2,\ldots,b_j) \in PRT_0[r]} \sum\limits_{(b_{m_{i_1}},  b_{m_{i_2}},  \ldots , b_{m_{i_j}} ) \in ARG[b] } 
	|D_{b_{m_{i_1}}}| |D_{b_{m_{i_2}}}| \cdots |D_{b_{m_{i_j}}}| \\
	&= \sum\limits_{j=1}^{n} \sum\limits_{I \in  \mathcal{I}_j^{j}} \sum\limits_{(b_1,b_2,\ldots,b_j) \in PRT_0[r]}  |D_{b_{m_{i_1}}}|  |D_{b_{m_{i_2}}}| \cdots |D_{b_{m_{i_j}}}| 
	|ARG[(b_1,b_2,\ldots,b_j)]|
\end{align*}
wherein, for an $I \in \mathcal{I}_{j}^{j}$, $\{i_1,i_2,\ldots,i_j\}$ is the set of maximal elements of $I$.
\par By Theorem \ref{r=Mw}, for $ r =t M_{w} + a$; $1 \leq t \leq n-1$, $0 \leq a \leq M_w$, we can get, 
\begin{align*}
	|A_{tM_w + a}|  &= \sum\limits_{i=0}^{t} \sum\limits_{j=1}^{n} \sum\limits_{I \in  \mathcal{I}_{j}^{i+j}} \sum\limits_{(b_1,b_2\ldots,b_j) \in PRT_{i}[r]}  
	\sum\limits_{(b_{m_{i_1}}, b_{m_{i_2}},   \ldots , b_{m_{i_j}} ) \in ARG[b] }  |D_{b_{m_{i_1}}}| \\& \hspace*{7cm} |D_{b_{m_{i_2}}}| \cdots |D_{b_{m_{i_j}}}|  q^{i} \\ 	&
	= \sum\limits_{i=0}^{t} \sum\limits_{j=1}^{n} \sum\limits_{I \in  \mathcal{I}_{j}^{i+j}} \sum\limits_{(b_1,b_2,\ldots,b_j) \in PRT_{i}[r]} |D_{b_{m_{i_1}}}|  |D_{b_{m_{i_2}}}|  \cdots |D_{b_{m_{i_j}}}|   q^{i} \\& \hspace*{7cm}
	|ARG[(b_1,b_2,\ldots,b_j)]|
\end{align*}
\par Alternatively,  from Theorem \ref{Ar without ARG},  we can determine  the cardinality of $A_r$ without finding $ARG[b]$, as $k_i=1 ~\forall ~ i \in [n]$. 
\begin{corollary} \label{Ar for pw dist without ARG}
	For any $1 \leq r \leq n M_{w} $, the number of $n$-tuples  $u \in \mathbb{F}_{q}^n$ having $w_{(P,w)}(u)= r$  is 
	\begin{align*}
		|A_{r}| =	\sum\limits_{i=1}^{n} \sum\limits_{j=1}^{i} \sum\limits_{I \in \mathcal{I}_j^{i}}  \sum\limits_{ (b_1,b_2,\ldots,b_j)  \in PRT _{i-j} [r] } q^{i-j} 
		\prod\limits_{s=1}^{l}  {|D_{t_s}|}^{r_s} {j - (r_1+r_2+\ldots+r_{s-1})\choose r_s} 
	\end{align*}
	where  $r_s$ parts among the $j$ parts in  $b_1, b_2, \dots, b_j  $  are equal to $t_s$  ($1 \leq s \leq l$).
\end{corollary}
\begin{corollary}
	The number of $n$-tuples  $ x \in \mathbb{F}_{q}^n$ having the $(P,w,\pi) $-weight $nM_w$  is 	$ {|D_{M_w}|}^{t} q^{n-t}  $ where $t$ is the  number of maximal elements of the given poset $P$. 
\end{corollary}
If $P=([n],\preceq)$ is a chain (here $1 \preceq 2 \preceq \ldots \preceq n$), each ideal $I \in \mathcal{I}(P)$  will have only one maximal element.  Let  $v \in \mathbb{F}_q^n$,  $I_v^P=  \langle supp (v) \rangle $ and $j$ be the maximal element of $I_v^P$. Then, 
$	w_{(P,w)}(v)= w{(v_j)} + (|I_v^P|-1) M_w $.   Moreover, 
\begin{itemize}
	\item   For $1 \leq r \leq M_{w} $, we have $v \in A_r$ iff $| \langle supp (v) \rangle | = 1 $, and hence $	|A_{r}| = |D_{r}|$.
	\item	For $r= tM_{w}+ a $, where $0 < a \leq M_{w}$ and $1 \leq t \leq n-1$, we have $v \in A_r$ iff $| \langle supp (v) \rangle | = t + 1 $. Thus,
	$|A_{tM_{w} +  a }| =  q^{t} |D_{a}|$.
\end{itemize} 
\subsection{$(P,\pi)$-space}
Given a poset $P=([n],\preceq)$, a label map $\pi$ on $[n]$, and considering  $w$ to be the Hamming weight on  $ \mathbb{F}_{q}$,  the $(P,w,\pi)$-space becomes poset block or  $(P,\pi)$-space \cite{Ebc} (see Remark \ref{becomes}). Now $M_w= 1$, and hence we have
$|D_0| = 1$, $|D_1| =  q-1$ and $|D_1^{k_j}| = q ^{k_j}- 1$ for $1 \leq j \leq n$.  
Clearly,  $x$ $\in A_r$ iff $ | \langle supp_{\pi} (x) \rangle | =r $.  
\par   As $M_w= 1$, we have $|PRT_{i-j}[r]|=1$  for any $r \leq n$, for      $PRT_{i-j}[r]=\{(1,1,\ldots,1) : 1+1+\cdots+1=r-(i-j)\}$.   
From Theorem  \ref{r=Mw},   the number of $N$-tuples  $ x \in \mathbb{F}_{q}^N$ having $w_{(P,\pi)}(x)= r$  is
\begin{align*}
	\sum\limits_{i=0}^{r-1} \sum\limits_{j=1}^{n} \sum\limits_{I \in  \mathcal{I}_j^{j+i}} \sum\limits_{(1,1,\ldots,1) \in PRT_{i}[r]}   |D_{1}^{k_{i_1}}|  |D_{1}^{k_{i_2}}| \cdots |D_{1}^{k_{i_j}}|  q^{k_{l_1} + k_{l_2} + \cdots + k_{l_i}} 
\end{align*}
Since $x$ $\in A_r$ iff $  I_x^P \in \mathcal{I}_j^{r} $ for some $ j \leq r$, we have
\begin{align*}
	|A_{r}| =   \sum\limits_{j=1}^{r} \sum\limits_{I \in  \mathcal{I}_j^{r}}  (q^{k_{i_1}}-1) (q^{k_{i_2}}-1)  \cdots (q^{k_{i_j}}-1)    q^{k_{l_1} + k_{l_2} + \cdots + k_{l_{r-j}}} 
\end{align*}
where $ Max(I) = \{i_1,i_2,\ldots,i_j\} $ and $ I  \setminus Max(I) = \{l_1, l_2, \dots, l_{r-j} \}$ for a given  $I  \in \mathcal{I}_j^{r}$.
Moreover, if $k_i = k$ for all $i \in [n]$, then 
$
|A_{r}| =   \sum\limits_{j=1}^{r} \sum\limits_{I \in  \mathcal{I}_j^{r}}  (q^{k}-1)^{j}   q^{k(r-j)} . $

\begin{corollary}
	Let	$k_i = k$ for all $i \in [n]$.	Then, $x \in A_n$ iff  $I_x^{P,\pi}  \in  \mathcal{I}_t^{n} $ where $t$ is the  number of maximal elements of the given poset $P$. Moreover,
	$ | A_n | =  (q^{k}-1)^{t} q^{k(n-t)}  $. 
\end{corollary}
If  $P=([n],\preceq)$ is  a chain, there is only one  maximal element in any ideal $I \in \mathcal{I}(P)$ and $|\mathcal{I}_1^{r}|=1$ for $1 \leq r \leq n $. Let  $v \in \mathbb{F}_q^n$ and  $I_v^{P,\pi}= \langle supp_{\pi}(v) \rangle $ with maximal element $j$. Then, 
\begin{align*}
	w_{(P,w,\pi)}(v) &= \tilde{w}^{k_j}{(v_j)} + (|I_v^{P,\pi}|-1) M_w  \\& = |I_v^{P,\pi}| \\ &= |\langle j \rangle| 
	\\& = w_{(P,\pi)}(v).
\end{align*}
This weight coincides with the RT-weight if $k_i=1$  for each $i \in [n]$. 	  
We have $x$ $\in A_r$ iff $ | \langle supp_{\pi} (x) \rangle | =r $, and thus 
\begin{align*}
	|A_{r}| = (q^{k_r}-1)  q^{k_1+k_2+\dots+k_{r-1}} \text{ for } 1 \leq r \leq n . 
\end{align*}
\par Moreover, if $k_i = k$ for each $i \in [n]$, then $
|A_{r}| = (q^{k}-1)  q^{k({r-1})} \text{ for }1  \leq r \leq n .$
\subsection{ $\pi$-space}
If $P=([n],\preceq)$ is an anti-chain and $w$ is the Hamming weight on $\mathbb{F}_q$, the $(P,w,\pi)$-space becomes the classical $( \mathbb{F}_q^N,d_{\pi})$-space \cite{fxh} (see Remark \ref{becomes}).
Now $A_r = \{x= x_{1} \oplus x_{2} \oplus \ldots \oplus x_{s} \in \mathbb{F}_{q}^{N}  : w_{\pi}(x)=r  \}$ for each $1 \leq r \leq n$. Clearly,  $x$ $\in A_r$ iff $ | supp_{\pi} (x) | =r $.  
\par Now each subset of $[n] $ is an ideal. Let $P_r$ be the collection of all subsets $B$ of $[n]$ such that $|B|=r$. If $B=\{i_1,i_2,\ldots,i_r\} \subseteq [n]$ then  
\begin{align*}
	| A_r | =  \sum\limits_{  B \in P_r} (q^{k_{i_1}}-1)(q^{k_{i_2}}-1)\cdots(q^{k_{i_r}}-1) .
\end{align*}
\par 	If $k_i=k$ for each $i \in [n]$, we have 
\begin{align*}
	| A_r | =  {n \choose r} (q^{k}-1)^r \text{ for } 1 \leq r \leq n .  
\end{align*}
\subsection{ $P$-space}
The $(P,w,\pi)$-space becomes the classical $(\mathbb{F}_{q}^{n}, d_{P})$-poset space \cite{Bru}  by considering $w$ 
as the Hamming weight on $\mathbb{F}_q$ and $k_i=1$ $\forall$ $i \in [n]$ (see Remark \ref{becomes}). 
As $M_w=1$,  $\tilde{w}^{k_s}(x_s)=w(x_s)=1$  for all $x_s \in \mathbb{F}_q^{k_s}\setminus \{0\}$, and  $|PRT_{i-j}[r]|=1$  for any $r \leq n$  so that $ARG[b]=PRT_{i-j}[r]$.  
\par  For  each $ 1 \leq r \leq n$, if  $A_r = \{x= (x_{1},  x_{2}, \ldots,  x_{n}) \in \mathbb{F}_{q}^{n}  : w_{P}(x)=r  \}$, then  $x$ $\in A_r$ iff $ | \langle supp (x) \rangle | =r $. Therefore, 
\begin{align*}
	|A_{r}| = \sum\limits_{j=1}^{r} \sum\limits_{I \in  \mathcal{I}_j^{r}}  ( q-1 )^{j}  q^{r-j}  . 
\end{align*}		
\par	In particular,   $ A_n= (q-1)^{t} q^{n-t}$ where $t$ is the number of maximal elements in  $P$. If  $P=([n],\preceq)$ is a chain, then 
\begin{align*}
	|A_{r}| ={q}^{r-1} {(q-1)}   \text{ for } 1 \leq r \leq n .
\end{align*}		
\section{Weighted Coordinates Hierarchical Poset Block Spaces} 
\label{section:5}
Let $P=([n],\preceq)$ be a poset on $[n]$. 
The height $h(i)$ of an element $i \in P $ is the cardinality of the largest chain having $i$ as the maximal element. The height $h(P)$ of the poset $P$ is the maximal height of its elements,
i.e., $h(P) = \max \{h(i) : i \in [n]\}$. The $i$-th level $\Gamma_i^P  $ of a poset 
$P$ is the set of all elements with height $i$, i.e.,
\begin{align*}
	\Gamma_i^P  = \{ j \in [n] : h(j)=i \}.
\end{align*}
\par Let $|\Gamma_i^P |=n_i$ so that and $\dot\bigcup_{i=1}^{h(P)} \Gamma_i^P  = [n]$ and $n=n_1+n_2+\ldots+n_{h(P)}$. Clearly, each level of a poset is an antichain and the number of levels in a poset $P$ is $h(P)$.
\par  A poset $P = ([n],\preceq)$ is said to be \textit{hierarchical} if elements at different levels are comparable, i.e., if $ t_i \in \Gamma_i^P $ and $ t_j \in \Gamma_j^P$ with $i \neq j$  then $t_i \preceq t_j$ if and only if $i \preceq j$.
\par  Thus, when $P$ is an hierarchical poset on $[n]$ with $h$ levels, the following hold:
\begin{itemize}
	\item If $I$ is an ideal of $P$ then  $\max(I) \subseteq 	\Gamma_i^P  $ for some $i$.
	\item  If $I \in \mathcal{I}^{t}$ and $t = n_1 + n_2 + \ldots +n_{j-1}+l$ where $1 \leq l \leq n_{j}$ then $|\max(I)| =l$.
	\item  If $i = n_1 + n_2 + \ldots +n_{j}$ where $1 \leq j \leq h$, then $| \mathcal{I}^i |=1  $.
\end{itemize}

\begin{proposition}\label{n choose l}
	Let $t = n_1 + n_2 + \ldots +n_{j-1}+l$ where $1 \leq l \leq n_{j}$. Then 
	$	|\mathcal{I}_l^{t}| = {n_{j} \choose l}$.
\end{proposition}
\begin{proof}
	Let $I \in \mathcal{I}_l^{t}$. Then, $\max(I) \subseteq \Gamma_j^P $ and $|\max(I)| =l$.
	Thus, $| \mathcal{I}_l^{t}|$ is nothing but the number of ways $l$ maximal elements can be chosen from $\Gamma_j^P$. As $|\Gamma_j^P|=n_j$, the proof follows.
\end{proof}
\par Now, let $x =  x_1 \oplus x_2 \oplus \cdots \oplus x_n \in  \mathbb{F}_{q}^{N} $ where $x_i \in  \mathbb{F}_{q}^{k_i} $. As $P$ is hierarchical with $h$ levels, $M_x^{P,\pi} \subseteq \Gamma_i^P $ for some $i \in \{1,2,\ldots,h\}$. The  $(P,w,\pi)$-weight of  $x \in \mathbb{F}_q^N$ is
\begin{equation*}
	w_{(P,w,\pi)}(x) := \sum\limits_{j \in M_x^{P,\pi} \subseteq \Gamma_i^P} \tilde{w}^{k_j}{(x_j)} + (n_1+n_2+\ldots+n_{i-1}) M_w
\end{equation*} 
where $n_0=0$. 
This weight can be called as $(P,w,\pi)$-hierarchical weight and the corresponding $(P,w,\pi)$-metric $d_{(P,w,\pi)}$ is called as the $(P,w,\pi)$-hierarchical metric on $\mathbb{F}_{q}^{N} $.   We call the space $ (\mathbb{F}_{q}^N, ~d_{(P,w,\pi)} )$ with  hierarchical poset $P$ as  \textit{weighted coordinates hierarchical poset block space} or \textit{$(P,w,\pi)$-hierarchical block space}. 
\begin{proposition}
	Let $r= tM_w + s_0$ where $t = \sum\limits_{i=1}^{j-1} n_i$ and $s_0 \in \{1,2,\ldots,n_jM_w\}$. If $x \in A_r$ then $M_x^{P,\pi} \subseteq \Gamma_j^P$.
\end{proposition}
\par The results obtained in Section \ref{section:4} can also be applied to hierarchical poset. Let $n_0=0$.
Thus, the weight distribution of the weighted coordinates hierarchical poset block space $ (\mathbb{F}_{q}^N, ~d_{(P,w,\pi)} )$ can be deduced from Theorem \ref{r=Mw} and Theorem \ref{Ar without ARG} which is stated as the following results:
\begin{theorem}\label{r=Mw Hierarchical}
	For any $ r =t M_{w} + s_0$ with  $t = n_1 + n_2 + \ldots +n_{j-1}$ and $0 < s_0 \leq n_jM_w$, the number of $N$-tuples  $ x \in \mathbb{F}_{q}^N$ having $w_{(P,w,\pi)}(x)= r$  is
	\begin{align*}
		|A_{tM_w + s_0}| =	\sum\limits_{j=1}^{n_j}   \sum\limits_{I \in  \mathcal{I}_j^{j+t}}
		\sum\limits_{(b_1,b_2,\ldots,b_j) \in PRT_{t}[r]}  \sum\limits_{(b_{m_{i_1}},b_{m_{i_2}},  \ldots , b_{m_{i_j}}) \in ARG[(b_1,b_2,\ldots,b_j)] }  |D_{b_{m_{i_1}}}^{k_{i_1}}| \\  |D_{b_{m_{i_2}}}^{k_{i_2}}|  \cdots |D_{b_{m_{i_j}}}^{k_{i_j}}|   q^{k_{l_1} + k_{l_2} + \cdots + k_{l_t}} 
	\end{align*} 
	wherein, for an $I \in {\mathcal{I}_j^{j+t} }$,  $Max(I)=\{i_1,i_2,\ldots,i_j\}$ and  $I \setminus Max(I) = \{l_1, l_2, \dots, l_{t} \}$.
\end{theorem}

\begin{theorem}\label{Ar without ARG Heirarchical}
	Let $k_i=k$ $\forall$ $i \in [n]$ and $t = n_1 + n_2 + \ldots +n_{j-1}$. For $r=  tM_w + s_0$ where $s_0 \in \{1,2,\ldots,n_jM_w\}$, the number of $N$-tuples  $ x \in \mathbb{F}_{q}^N$ having $w_{(P,w,\pi)}(x)= r$ is 
	\begin{align*}
		|A_{r}| &=	\sum\limits_{j=1}^{n_j} \sum\limits_{I \in \mathcal{I}_j^{j+t}}  \sum\limits_{(b_1,b_2,\ldots,b_j) \in PRT_{t}[r]} q^{k(t-j)} 
		\prod\limits_{s=1}^{l} |D_{t_s}^k | ^{r_s}     {j - (r_1+r_2+\ldots+r_{s-1})\choose r_s} 
	\end{align*}
	where  $r_s$ parts among the $j$ parts in  $b_1, b_2, \dots, b_j  $  are equal to $t_s$ 
	$(1 \leq s \leq l)$.
\end{theorem}

\begin{corollary} \label{wpirball hierarchical}
	The number of $N$-tuples in a  $(P,w,\pi)$-ball of radius $r$ centered at $ x \in \mathbb{F}_{q}^N $ is  $| B_{r}(x) | = 1 + \sum\limits_{t=1}^{r} 	|A_t|$. 
\end{corollary} 
\subsection{Hierarchical poset block space}
In this subsection, we consider $w$ as the Hamming weight on $\mathbb{F}_q$ and $P$ as the hierarchical poset of height $h$ on $[n]$. Then for any $x \in \mathbb{F}_q^N$, $M_x^{P,\pi} \subseteq \Gamma_i^P$ for some $i \in  \{1,2,\ldots,h\}$ and thus the $(P,w,\pi)$-hierarchical weight of  $x \in \mathbb{F}_q^N$ is
\begin{align*}
	w_{(P,w,\pi)}(x) &:= \sum\limits_{j \in M_x^{P,\pi} \subseteq \Gamma_i^P} w_H{(x_j)} + (n_1+n_2+\ldots+n_{i-1}) \\
	&=|M_x^{P,\pi} |+n_1+n_2+\ldots+n_{i-1}
\end{align*}
which is the $(P,\pi)$-hierarchical poset block weight of $x$.
If $k_i=1$ $\forall$ $i \in [n]$, it becomes the \textit{hierarchical poset weight} as defined in \cite{hierarchical}. Hence, the $(P,w,\pi)$-space with $P$ as an hierarchical poset generalizes the \textit{hierarchical poset space described in \cite{hierarchical}}.        
\par Now, we give the complete weight distribution for hierarchical poset block metric space.  

\begin{theorem}
	For $ r = n_1 + n_2 + \ldots +n_{j-1}+a$ where $1 \leq a \leq n_{j}$, the number of $N$-tuples  $ x \in \mathbb{F}_{q}^N$ having $w_{(P,w,\pi)}(x)= r$  is
	\begin{align*}
		|A_{r}| =	\sum\limits_{I \in  \mathcal{I}_{a}^{r}}  {(q^{k_{i_1}}-1)}   {(q^{k_{i_2}}-1)} \cdots {(q^{k_{i_a}}-1)}  
		q^{k_{l_1} + k_{l_2} + \cdots + k_{l_{r-a}}} 
	\end{align*} 
	wherein, for an $I \in {\mathcal{I}_{a}^{r} }$,  $Max(I)=\{i_1,i_2,\ldots,i_{a}\}$ and  $I \setminus Max(I) = \{l_1, l_2, \dots, l_{r-a} \}$. In particular, if $k_i=k$ $\forall$ $i \in [n]$ then $
	|A_{r}| =	\sum\limits_{I \in  \mathcal{I}_{a}^{r}}  {(q^k-1)}^{a}  q^{k{(r-a)}} $.		
\end{theorem}
\begin{remark}
	In fact, $|A_r| = {n_j \choose a} {(q^k-1)}^{a}  q^{k{(r-a)}}$ when $k_i=k$ $\forall$ $i \in [n]$ by Proposition \ref{n choose l}. 
\end{remark}
\begin{corollary}
	For any $ r = t+a$  with $ t = n_1 + n_2 + \ldots +n_{j-1}$ and $1 \leq a \leq n_{j}$, the number of $N$-tuples in a $(P,w,\pi)$-ball of radius $r$ centered at $ x \in \mathbb{F}_{q}^N $ is  
	\begin{align*}
		| B_{r}(x) | =  q^{k_{l_1} + k_{l_2} + \cdots + k_{l_{r-a}}} \bigg(1 + \sum\limits_{l=1}^{a}  \sum\limits_{I \in  \mathcal{I}_{l}^{t+l}}   {(q^{k_{i_1}}-1)}   {(q^{k_{i_2}}-1)} \cdots {(q^{k_{i_l}}-1)}  \bigg).
	\end{align*}
	In particular, if $k_i=k$ $\forall$ $i \in [n]$ then 
	$| B_{r}(x) | =  q^{k(r-a)} \big(1 + \sum\limits_{l=1}^{a}  \sum\limits_{I \in  \mathcal{I}_{l}^{t+l}}   (q^k-1)^l  \big)$.
\end{corollary}	
\begin{proof}
	Let $ t = n_1 + n_2 + \ldots +n_{j-1}$. As $|\mathcal{I}^t|=1$, let $J \in \mathcal{I}^t$.
	Now $r=t+a$ and  thus for any   $I \in {\mathcal{I}_{a}^{r} }$,  we have $J \subseteq I$ and $J=I \setminus Max(I)$.	Since 
	\begin{align*}
		| B_{r}(x) | = (1 + \sum\limits_{i=1}^{t} |A_i| ) + \sum\limits_{l=1}^{a} |A_{t+l}|  \text{ and }  1 + \sum\limits_{i=1}^{t} |A_i| = \prod\limits_{i \in J}q^{k_i},
	\end{align*}
	we have
	\begin{align*}
		| B_{r}(x) | =  q^{k_{l_1} + k_{l_2} + \cdots + k_{l_{r-a}}} \bigg(1 + \sum\limits_{l=1}^{a}  \sum\limits_{I \in  \mathcal{I}_{l}^{t+l}}   {(q^{k_{i_1}}-1)}  {(q^{k_{i_2}}-1)} \cdots {(q^{k_{i_{l}}}-1)} \bigg). 
	\end{align*}
\end{proof}
\subsection{Hierarchical poset space}
The $(P,w,\pi)$-space becomes the hierarchical poset space of \cite{hierarchical} if $P$ is a  hierarchical, $w$ is the Hamming weight on $\mathbb{F}_q$ and $\pi(i)=1$ $\forall$ $i \in [n]$ (see Remark \ref{becomes}).   	For $ r = n_1 + n_2 + \ldots +n_{j-1}+a$ where $1 \leq a \leq n_{j}$, we have 
\begin{align*}
	|A_{r}| = 
	\sum\limits_{I \in  \mathcal{I}_{a}^{r}}  {(q-1)}^{a}  q^{r-a} = {n_j \choose a} (q-1)^{a}  q^{r-a}.
\end{align*}
Further, the cardinality of an $r$-ball centered at $ x \in \mathbb{F}_{q}^n $ is  
\begin{align*}
	| B_{r}(x) | = 	 q^{r-a} \bigg(1 + \sum\limits_{l=1}^{a} {n_j \choose l} (q-1)^{l} \bigg).
\end{align*}
\section{$I$-Perfect block codes and $t$-perfect block codes}\label{perfect codes in P,w,pi-space}
In this section, we first recall the balls in poset block space and then we define it for the $(P,w,\pi)$-space. We investigated the characteristics of $I$-perfect codes for an ideal $I$ and discussed the relationship with $t$-perfect codes and MDS $(P,w,\pi)$-codes.
Moreover, the duality theorem with respect to $(P,w,\pi)$-metric is obtained. 
\par   For an ideal (or subset) $I$ in $P$, the $I$-ball centered at $u$ with respect to $(P,\pi)$-metric is $B_{I,(P,\pi)}(u) \triangleq \{ v \in  \mathbb{F}_q^N :  supp_\pi(u- v) \subseteq I \}$. If $I$ is an ideal and $v \in B_{I,(P,\pi)}(u) $ then $\langle supp_\pi(u- v) \rangle \subseteq I$. For each ideal $I$ in $P$, $B_{I,(P,\pi)}(0)$ is a linear subspace of $\mathbb{F}_q^N$ and the space $\mathbb{F}_q^N$ can be partitioned into $I$-balls. Also, $|B_{I,(P,\pi)}{(u)}|=q^{\sum_{i\in I} k_i}$.
For $u \in \mathbb{F}_q^N$, a 
$t$-ball centered at $u$ with respect to $(P,\pi)$-metric is  $B_{(P,\pi)}(u,t) \triangleq \{ v \in \mathbb{F}_q^N :  d_{(P,\pi)}(u- v)  \leq t \}$ where $ t \leq n$. Every $t$-ball is the union of all $I$-balls where $I \in \mathcal{I}^t$, i.e. $	B_{(P,\pi)}(u,t)= \bigcup\limits_{I \in  \mathcal{I}^t}B_{I,(P,\pi)}(u)$. For more details, one can see \cite{bkdnsr}. 
\par \sloppy  The $(P,w,\pi)$-ball centered at a point $y \in \mathbb{F}_q^N$ with radius $r$ is the set $ B_{(P,w,\pi)}(y,r)=\{x \in \mathbb{F}_{q}^{N} : d_{(P,w,\pi)}(y,x) \leq r\}$.
The $(P,w,\pi)$-sphere centered at $y$  with radius $r$ is  $S_{(P,w,\pi)}(y,r)=\{x \in \mathbb{F}_{q}^{N} : d_{(P,w,\pi)}(y,x) = r\}$. As the $d_{(P,w,\pi)}$-metric is translation invariant, we have  $B_{(P,w,\pi)}(y,r) = y + B_{(P,w,\pi)}(0,r)$ where $0 \in \mathbb{F}_q^N$. Clearly, $ | B_{(P,w,\pi)}(0,r) | =  1 + \sum\limits_{t=1}^{r} | S_{(P,w,\pi)}(0,t) |$ and $ | S_{(P,w,\pi)}(0,t) |$ is equal to the number of $x \in \mathbb{F}_{q}^{N}$ such that $w_{(P,w,\pi)}(x)=t$. 
For $y \in \mathbb{F}_q^N$, let $r= tM_{w}+ s $, $0 < s \leq M_{w}$ and $1 \leq t \leq n-1$. If $v \in B_{(P,\pi)}(u, t+1) $ then $ supp_\pi(u-v) \subseteq I$ for an ideal $I \in \mathcal{I}^{t+1}$. Thus, $ w_{(P,w,\pi)}(u -v) \leq (t+1) M_w$ and
$v \in B_{(P,w,\pi)}(u, (t+1) M_w) $. Hence, we have
\begin{theorem}\label{pwpiChept-I-Bp inside Bp,w}
	Let $w$ be a weight on $\mathbb{F}_{q} $, $r=tM_w+s$, $0 < s \leq M_{w}$ and $1 \leq t \leq n-1$. Then, for any $u \in \mathbb{F}_q^N $, we have 
	$	B_{(P,\pi)}(u, t+1) \subseteq B_{(P,w,\pi)}(u, (t+1) M_w)$  and  $ B_{(P,\pi)}(u, t) \subseteq B_{(P,w,\pi)}(u, r)$.
\end{theorem}  
\par  Now, with respect to $(P,w,\pi)$-metric, for an ideal (or subset) $I$ in $P$, we define the $I$-ball centered at $u \in \mathbb{F}_{q}^N$ as $B_{I,(P,w,\pi)}(u) \triangleq \{ v \in \mathbb{F}_q^N :  supp_\pi(u- v)  \subseteq I \}$.  As the $I$-ball is independent of the weights of the coordinates, it is same as the $I$-ball with respect to $(P,\pi)$-metric defined in \cite{bkdnsr,hkmdspc} with similar properties. Thus, $B_{I,(P,w,\pi)}(u) = B_{I,(P,\pi)}(u)$. 
In both the metrics, we shall denote an $I$-ball centered at  $u$ as $B_{I}(u)$ (further, $B_I$ denotes the $I$-ball centered at zero).
But, there are differences in characterization of MDS block codes in terms of $I$-perfect block codes as we will see in the subsequent results.  
\par 	A $(P,w,\pi)$-code $\mathcal{C}$ of length $N$ over $\mathbb{F}_q$  is said to be $I$-perfect if the  $I$-balls centered at the codewords of $\mathcal{C}$ are pairwise disjoint and their union covers the entire space $\mathbb{F}_q^N$.
It yields the following Lemma:  
\begin{lemma}\label{pwpiChept-I-rperfect to Iperf}
	Let $I$ be an ideal in $\mathcal{I}(P)$ and let $\mathcal{C}$ be a $k$-dimensional $(P,w,\pi)$-code of length $N$ over $\mathbb{F}_q$. Then the following are equivalent:
	\begin{enumerate}[label=(\roman*)]
		\item   $\mathcal{C}$ is an $I$-perfect code
		\item  $\sum_{i \in I} k_i= N-k$ (the covering condition)	and  $ | B_I \cap \mathcal{C} | =1 $ (the packing condition)
		\item 	$|B_I(x) \cap \mathcal{C}|=1 $ for all $x \in \mathbb{F}_q^N$; it means that each $n$-tuple of $\mathbb{F}_q^N$ belongs to exactly one $I$-ball centered at a codeword of $\mathcal{C}$.
	\end{enumerate}
\end{lemma}
\begin{theorem}
	For each ideal $I$ in $\mathcal{I}(P)$, there exists an $I$-perfect code with respect to both $(P,w,\pi)$-metric and $(P,\pi)$-metric.
\end{theorem}  
\begin{proof}
	Let $I \in \mathcal{I}(P)$. As $I$-balls are same in both $(P,w,\pi)$-metric and $(P,\pi)$-metric, 
	$\mathbb{F}_q^N$ can be partitioned into $I$-balls in both metric spaces.
	If $t$ is the number of $I$-balls in this partition then $t = q^{N -\sum_{i \in I} k_i}$. One can construct a code $\mathcal{C}$ of length $N$ over $\mathbb{F}_q$ with cardinality $t$ by choosing one $N$-tuple from each $I$-ball. Now, the  $I$-balls with respect to $(P,w,\pi)$-metric (or $(P,\pi)$-metric) centered at codewords of $\mathcal{C}$ are disjoint and cover $\mathbb{F}_q^N$. Hence,  $\mathcal{C}$ is an $I$-perfect $(P,w,\pi)$-code $\mathcal{C}$ (or $(P,\pi)$-code).  
\end{proof} 

\par  If we consider $k_i=1$ $\forall$ $i \in [n]$, we recover	the result in  \cite{aks} as a corollary.
\begin{corollary}
	For each ideal $I$ in $\mathcal{I}(P)$, there exists an $I$-perfect code with respect to both $(P,w)$-metric and $P$-metric. 
\end{corollary}  
Now, the following results are immediate:
\begin{theorem}\label{pwpiChept-I-pwtoP perfect}
	Let $\mathcal{C} $ be a code of length $N$ over $\mathbb{F}_q$. Then $\mathcal{C} $ is $I$-perfect with respect to $(P,w,\pi)$-metric if and only if $\mathcal{C}$ is $I$-perfect with respect to $(P,\pi)$-metric. 
\end{theorem}
\begin{corollary}
	Let $\mathcal{C} $ be a code of length $n$ over $\mathbb{F}_q$. Then $\mathcal{C} $ is $I$-perfect with respect to $(P,w)$-metric if and only if $\mathcal{C}$ is $I$-perfect with respect to $P$-metric. 
\end{corollary}
A $(P,w,\pi)$-code $\mathcal{C}$ of length $N$ over $\mathbb{F}_q$  is said to be an $r$-error correcting code if the  $r$-balls centered at the codewords of $\mathcal{C}$ are pairwise disjoint.
It is said to be $r$-perfect if the  $r$-balls centered at the codewords of $\mathcal{C}$ are pairwise disjoint and their union covers the entire space $\mathbb{F}_q^N$.
\begin{theorem}
	Let $\mathcal{C}$ be a $(P,w,\pi)$-code of length $N$ over $\mathbb{F}_q$. If $\mathcal{C}$ is a $tM_w$-error correcting code then
	for any two distinct codewords $x,y \in \mathcal{C}$, $x-y \notin B_{I \cup J}$ $\forall$  $I, J \in \mathcal{I}^{t}$. 
\end{theorem}
\begin{proof}
	Let $r=tM_w$ and $\mathcal{C}$ be an $r$-error correcting $(P,w,\pi)$-code. Let $x,y \in \mathcal{C}$ such that $x \neq y$. Suppose that $x-y \in B_{I \cup J}$ for some  $I, J \in \mathcal{I}^t$. Choose $ z \in \mathbb{F}_q^N$ such that $ z= x - (x-y)_{I \setminus J} $ where
	$(x-y)_{I \setminus J}$ denotes the $N$-tuple vector in which $x_j-y_j=0$ $\forall$ $j \in I \cap J $. 
	Then, $d_{(P,w,\pi)}(z,x) \leq | \langle supp_\pi (x-z) \rangle | M_w= | \langle supp_\pi ( x-x+(x-y)_{I \setminus J} ) \rangle| M_w   \leq  |I | M_w=t M_w$. Thus, $z \in B_{(P,w,\pi)}(x,r)$. Now, $d_{(P,w,\pi)}(z,y) \leq  | \langle supp_\pi ( (x-y)- (x-y)_{I \setminus J} ) \rangle|M_w \leq  tM_w $, so that $z \in B_{(P,w,\pi)}(y,r)$ as well. Therefore, $\mathcal{C}$ is not an $r$-error-correcting code, which contradicts the assumption.
\end{proof}
\begin{remark}\label{pwpiChept-I-terror}
	For a particular case when $w$ is the Hamming weight (that is, with respect to poset (block) metric), one can see that the converse of above Theorem also holds which is described in 
	[ref. \cite{bkdnsr}, Theorem $4.1$]: if $\mathcal{C}$ is a $(P,\pi)$-code of length $N$ over $\mathbb{F}_q$, then $\mathcal{C}$ is a $t$-error correcting code if and only if for any two distinct codewords $x,y \in \mathcal{C}$, $x-y \notin B_{I \cup J}$ $\forall$  $I, J \in \mathcal{I}^t$. 
	Thus,  if $\mathcal{C} $ is a $tM_w$-error correcting code with respect to $(P,w,\pi)$-metric then $\mathcal{C}$ is a $t$-error correcting code with respect to $(P,\pi)$-metric.
\end{remark}
\par 
In the following example, we will see that there is a block code $\mathcal{C} $ which is MDS with respect to $(P,\pi)$-metric but $\mathcal{C} $ is not MDS with respect to $(P,w,\pi)$-metric. Moreover, 
$\mathcal{C} $ is $I$-perfect for all $I \in \mathcal{I}^t $, but  $\mathcal{C}$ is not a $tM_w$-perfect $(P,w,\pi)$-code when $|\mathcal{I}^t | \geq 2$ for some $t  \leq n$.
\begin{example}\label{pwpiChept-I-ex1}
	Let $\preceq$ to be the partial order relation on the set $[5]=\{1,2,3,4,5\}$ such that $ i \preceq 4$ when $i=1,2$ and $ 3 \preceq 5$. Consider $w$ as the Lee weight on $\mathbb{Z}_7$ and let $\mathbb{Z}_7^8 = \mathbb{Z}_7^3 \oplus \mathbb{Z}_7^2 \oplus \mathbb{Z}_7 \oplus \mathbb{Z}_7 \oplus \mathbb{Z}_7$ where  $k_1=3,k_2=2,k_3=1,k_4=1,k_5=1$.
	Let $\mathcal{C} = \{(0,0,0,0,0,0,a,a) \in \mathbb{Z}_7^8 : a \in \mathbb{Z}_7 \}$, a linear code of dimension $k=1$ over $\mathbb{Z}_7$.  
	We have $m_w=1$, $M_w=3$ and $d_{(P,\pi)}(\mathcal{C})=5$. Also, $N-k=7$, $d_{(P,w,\pi)}(\mathcal{C})=11$ and  $r_{\tilde{\omega}}=\big\lfloor \frac{d_{(P,w,\pi)} (\mathcal{C})-m_w}{M_w} \big\rfloor = 3$.
	\par Here, $\mathcal{I}^4 = \{I_1,I_2\}$ where  $I_1=\{1,2,3,4\}$ and $I_2=\{1,2,3,5\}$ are the ideals of cardinality $4$.
	Since $ \max_{J \in  \mathcal{I}^{d_{(P,\pi)}(\mathcal{C}) - 1}} \big\{\sum_{i \in J} k_{i}\big\} =7=N-k$, $\mathcal{C} $ is MDS with respect to $(P,\pi)$-metric.
	But, $\mathcal{I}^3 = \{I_3,I_4,I_5,I_6\}$ where  $I_3=\{1,2,3\}$, $I_4=\{1,2,4\}$,  $I_5=\{1,3,5\}$ and $I_6=\{2,3,5\}$ are the ideals of cardinality $3$.
	Here, 
	$ \max_{J \in  \mathcal{I}^{r_{\tilde{\omega}}}}    \big\{\sum_{i \in J} k_{i}\big\} =\max\{6,6,5,4\}=6 < N-k$, and thus $\mathcal{C} $ is not MDS with respect to $(P,w,\pi)$-metric. 
	\par Now, we consider $\mathcal{I}^4 = \{I_1,I_2\}$. Then, $I_1$-ball centered at $(0,0,0,0,0,0,a,a) \in \mathcal{C}$, $B_I(0,0,0,0,0,0,a,a)= \{(x_1,x_2,x_3,x_4,x_5,x_6,x_7,a) : x_i \in \mathbb{Z}_{7}, 1 \leq i \leq 8 \}$ and $I_2$-ball centered at $(0,0,0,0,0,0,a,a)$, $B_J(0,0,0,0,0,0,a,a)= \{(x_1,x_2,x_3,x_4,x_5,x_6,a,x_7) : x_i \in \mathbb{Z}_{7}, 1 \leq i \leq 7 \}$. One can see that $\mathcal{C}$ is $I$-perfect for all $I \in \mathcal{I}^{4}$.  Now, by  taking $r=4M_w=12$, we see a non-zero codeword $c= (0,0,0,0,0,0,1,1) \in B_{(P,w,\pi)}(0,12)$. As $0,c \in \mathcal{C}\cap B_{(P,w,\pi)}(0,12)$,  $\mathcal{C} $ is not $12$-perfect with respect to $(P,w,\pi)$-metric but $\mathcal{C}$ is $I$-perfect for all $I \in \mathcal{I}^{4}$. \qed
\end{example} 
Now, we shall establish certain relationships between $r$-perfect and $I$-perfect block codes in the sequel.
\begin{theorem}\label{pwpiChept-I-rperfect}
	Let $I$ be the unique ideal of $P$ with cardinality $t$. 
	Let $\mathcal{C} $ be a $(P,w,\pi)$-code of length $N$ over $\mathbb{F}_q$. 
	Then $\mathcal{C} $ is $tM_w$-perfect with respect to $(P,w,\pi)$-metric iff $\mathcal{C}$ is $I$-perfect. 
\end{theorem}
\begin{proof}
	As $ \mathcal{C} $ is $tM_w$-perfect,  $d_{(P,w,\pi)}(c_1,c_2) > tM_w$ $\forall$ $c_1,c_2 \in \mathcal{C} $. By Theorem \ref{pwpiChept-I-J insdie I minusMaxI},  $B_{(P,w,\pi)}(c,tM_w) = B_I(c)$ $\forall$ $c \in \mathcal{C} $ and  $| \langle supp_{\pi}(c_1-c_2) \rangle | > t$ $\forall$ $c_1,c_2 \in \mathcal{C} $. Thus, $\bigcup\limits_{c \in \mathcal{C}}B_I(c)=\bigcup\limits_{c \in \mathcal{C}}B_{(P,w,\pi)}(c,tM_w)=\mathbb{F}_q^N$. Hence,
	$ \mathcal{C} $ is $I$-perfect. 
	Conversely,  let $ \mathcal{C} $ be an $I$-perfect $(P,w,\pi)$-code. Then $| \langle supp_\pi(c_1-c_2) \rangle | > t$ $\forall$ $c_1,c_2 \in \mathcal{C} $ by Theorem \ref{pwpiChept-I-pwtoP perfect} and $\bigcup\limits_{c \in \mathcal{C}}B_I(c)= \mathbb{F}_q^N$. As  $B_{(P,w,\pi)}(c,tM_w) = B_I(c)$, it follows that  $\mathcal{C}$ is  $r$-perfect.
\end{proof}
\begin{corollary}\label{pwpiChept-I-rperfectcoro1}
	Let $I$ be the unique ideal of $P$ with cardinality $t$.  
	Let $\mathcal{C} $ be a $(P,w,\pi)$-code of length $N$ over $\mathbb{F}_q$. 
	Then $\mathcal{C} $ is $r$-perfect with respect to $(P,w,\pi)$-metric iff $\mathcal{C}$ is  $t$-perfect with respect to $(P,\pi)$-metric. 
\end{corollary}
\par For $r=(N-k)M_w$, a complete characterization of an $(N-k)M_w$-perfect $(P,w,\pi)$-code is given as: 
\begin{theorem}\label{pwpiChept-I-n-k/s perfect imply I-perfect}
	Let $\mathcal{C} $ be a linear $[N,k]$ $(P,w,\pi)$-code of length $N$ over $\mathbb{F}_q$. Then 
	$\mathcal{C}$ is $(N-k)M_w$-perfect iff $\mathcal{I}^{N-k}= \{I\}$ and  $\mathcal{C}$ is $I$-perfect. 
\end{theorem}
\begin{proof}
	Let $ \mathcal{C} $ be $(N-k)M_w$-perfect.  Suppose that $\mathcal{C}$ is not $I$-perfect for some $I \in \mathcal{I}^{N-k}$. Then there exist two distinct codewords  $b, c \in \mathcal{C} $ such that 
	$ \langle supp_\pi(b-c) \rangle \subseteq I$. Thus, $c \in B_{(P,w,\pi)}{(b, (N-k)M_w)}$ which is a contradiction, as $ \mathcal{C} $ is $(N-k)M_w$-perfect. Therefore, $\mathcal{C}$ is $I$-perfect for every $I \in \mathcal{I}^{N-k}$.
	Now, suppose that $\{I\} \subsetneq \mathcal{I}^{N-k}$. This implies $ |B_{(P,w,\pi)}{(b, (N-k)M_w)}| > |B_{I}{(b)}| = q^{N-k}$ for each $b \in \mathcal{C}$ which means $ \mathcal{C} $ is not $(N-k)M_w$-perfect, a contradiction. The converse follows from Theorem \ref{pwpiChept-I-rperfect}.
\end{proof}
\begin{corollary}
	Let $\mathcal{C} $ be a linear $[n,k]$ $(P,w)$-code of length $n$ over $\mathbb{F}_q$. Then 
	$\mathcal{C}$ is $(n-k)M_w$-perfect iff $\mathcal{I}^{n-k}= \{I\}$ and  $\mathcal{C}$ is $I$-perfect. 
\end{corollary}
\begin{remark}
	If $ |\mathcal{I}^{N-k}| \geq 2 $ then there does not exist any $(N-k)M_w$-perfect linear  $(P,w,\pi)$-code  $ \mathcal{C} $ of cardinality $q^k$ over  $ \mathbb{F}_{q}$.
\end{remark}
\par In the remaining part of this paper, whenever we consider the blocks of equal sizes (i.e. $k_i = s$ $\forall$ $ i \in [n]$) and if we choose a $(P,w, \pi)$-code of length $N = ns$ with cardinality $q^k$ for some $k > 0$, then $s$ and $k$ shall be chosen such that $s $ divides $k$. Then, an ideal $ I \in \mathcal{I}^{ n-\frac{k}{s} }(P)$ will be of cardinality $n-\frac{k}{s}$ where $\frac{k}{s}$ is an integer. 
\section{MDS and $I$-perfect $(P,w,\pi)$-codes with equal block  length}\label{MDS-P,w,pi-space}
\par In \cite{aks}, we established the Singleton bound for any $(P,w)$-code $\mathcal{C}$ of length $n$ over $\mathbb{F}_q$ and investigated the relation between MDS codes and $I$-perfect codes for any ideal $I$. In this section, we extend this to weighted coordinates poset block codes when
all the blocks are of same length (that is, $\pi(i)= k_i=s ~\forall ~ i \in [n]$). 
\par From [ref. \cite{bkdnsr}, Theorem $5.1$], an $[N, k]$-code  $\mathcal{C}$ is MDS in $(P,\pi)$-space with label map $\pi(i)=s ~\forall ~ i$ if and only if $\mathcal{C}$ is $I$-perfect for every $I \in \mathcal{I}^{n-\frac{k}{s}}$. However, this is not the case with $(P,w,\pi)$-space.  But, we are in a position to relate the MDS codes and $I$-perfect codes in $(P,w,\pi)$-space.
\begin{theorem} \label{pwpiChept-I-Iperfectmdscode}	 
	Let $\mathcal{C}$ be an $[N,k]$ $(P,w,\pi)$-code of length $N$ over $\mathbb{F}_q$. If $\mathcal{C} $ is MDS with respect to $(P,w,\pi)$-metric then $\mathcal{C}$ is $I$-perfect for all $I \in \mathcal{I}^{n - \frac{k}{s}}$.   
\end{theorem}
\begin{proof}
	Suppose $\mathcal{C}$ is not $I$-perfect for some $I \in \mathcal{I}^{n - \frac{k}{s}}$. Then, there exist two distinct codewords $ u , v \in \mathcal{C}$ such that $supp_\pi(u-v) \subseteq I$. Then, $d_{(P,w,\pi)}(u-v) \leq (n - \frac{k}{s})M_w$ which is a contradiction to $\mathcal{C}$ being MDS.
\end{proof}
\begin{corollary} 
	Let $\mathcal{C}$ be an $[n,k]$ $(P,w)$-code of length $n$ over $\mathbb{F}_q$. If $\mathcal{C} $ is MDS with respect to $(P,w)$-metric then $\mathcal{C}$ is $I$-perfect for all $I \in \mathcal{I}^{n - k}$.   
\end{corollary}
\par  In the following example, we will see that converse of the above Theorem \ref{pwpiChept-I-Iperfectmdscode} is not true. 
\begin{example}
	Let $\preceq$ be the partial order relation on the set $[5]=\{1,2,3,4,5\}$ such that $ i \preceq 4$ and $ i \preceq 5$ for $i=1,2,3$. 
	Here, $\mathcal{I}^4 = \{I,J\}$ where  $I=\{1,2,3,4\}$ and $J=\{1,2,3,5\}$ are the ideals of cardinality $4$. Consider $\mathbb{Z}_7^{10} \equiv \mathbb{Z}_7^2 \oplus \mathbb{Z}_7^2 \oplus \mathbb{Z}_7^2 \oplus \mathbb{Z}_7^2 \oplus \mathbb{Z}_7^2 $ with $k_i=2$ $\forall$ $i \in [7]$.
	Considering $w$ as the Lee weight on $\mathbb{Z}_7$ and $\mathcal{C} = \{(0,0,0,0,0,0,a,b,a,b) \in \mathbb{Z}_7^{10} : a,b \in \mathbb{Z}_7 \}$, a linear code of dimension $k=2$ over $\mathbb{Z}_7$, 
	we have $s=2$, $n-\frac{k}{s}=4$, $m_w=1$, $M_w=3$ and $d_{(P,\pi)}(\mathcal{C})=5=n-\frac{k}{s}+1$. So, $\mathcal{C} $ is MDS with respect to $(P,\pi)$-metric and hence, $\mathcal{C}$ is $I$-perfect for all $I \in \mathcal{I}^{n - \frac{k}{s}}$ by Theorem \ref{pwpiChept-I-Iperfectmdscode}. But,  $d_{(P,w,\pi)}(\mathcal{C})=11$ and  $r_{\tilde{\omega}}=\big\lfloor \frac{d_{(P,w,\pi)} (\mathcal{C})-m_w}{M_w} \big\rfloor = 3< n-\frac{k}{s}$. Therefore
	$\mathcal{C} $ is not MDS with respect to $(P,w,\pi)$-metric.
\end{example}
\par But if we consider $w$ as the Hamming weight (as in the case of poset block space), we recover	the result in  \cite{bkdnsr} as a corollary.
\begin{corollary}\label{pwpiChept-I-posetmdsiff}
	Let $\mathcal{C}$ be an $[N,k]$ $(P,w,\pi)$-code of length $N$ over $\mathbb{F}_q$. If $w$ is the Hamming weight on $\mathbb{F}_q$, then $\mathcal{C} $ is MDS if and only if $\mathcal{C}$ is $I$-perfect for all $I \in \mathcal{I}^{n-\frac{k}{s}}$.  
\end{corollary}
\begin{proof}
	Suppose that $\mathcal{C} $ is $I$-perfect for all $I \in \mathcal{I}^{n-\frac{k}{s}}$. Then,  $\mathcal{C} $ is MDS with respect to poset block metric. As $w$ is the Hamming weight on $\mathbb{F}_q$, $m_w=M_w=1$ and $d_{(P,w,\pi)} (\mathcal{C})=d_{(P,\pi)} (\mathcal{C})$. Thus,
	$\big\lfloor \frac{d_{(P,w,\pi)} (\mathcal{C})-m_w}{M_w} \big\rfloor =  n-\frac{k}{s}$ and hence, $\mathcal{C} $ is MDS. The converse follows from the Theorem \ref{pwpiChept-I-Iperfectmdscode}.
\end{proof}

\begin{proposition}\label{pwpiChept-I-n-kM_wperfectMDs}
	Let $\mathcal{C}$ be an $[N,k]$ $(P,w,\pi)$-code of length $N$ over $\mathbb{F}_q$ and $w$ be such that $m_w=1$. If $\mathcal{C}$ is  $(n-\frac{k}{s})M_w$-perfect then $\mathcal{C}$ is MDS.
\end{proposition}
\begin{proof}
	Let $\mathcal{C}$ be an $(N-k)M_w$-perfect $(P,w,\pi)$-code. Then $d_{(P,w,\pi)}(\mathcal{C}) > (n-\frac{k}{s})M_w$ and $\big\lfloor \frac{d_{(P,w,\pi)} (\mathcal{C})-1}{M_w} \big\rfloor \geq n-\frac{k}{s}$.  Hence, from Theorem \ref{pwpiChept-I-sbwpic}, $\mathcal{C}$ is MDS.
\end{proof}
\par As $d_{(P,w,\pi)} (\mathcal{C})-1 \geq d_{(P,w,\pi)} (\mathcal{C})-m_w$,  $\big\lfloor \frac{d_{(P,w,\pi)} (\mathcal{C})-1}{M_w} \big\rfloor \geq \big\lfloor \frac{d_{(P,w,\pi)} (\mathcal{C})-m_w}{M_w} \big\rfloor $.  
If $\mathcal{C}$ is an MDS $(P,w,\pi)$-code for any weight function $w$ on $\mathbb{F}_q$ such that $m_w > 1$ and $k_i=s ~\forall~i $, then  $\mathcal{C}$ is also an MDS $(P,w,\pi)$-code for any weight function $w$ on $\mathbb{F}_q$ with $m_w=1$. 
\par If $w$ is a weight on $\mathbb{F}_q$ such that $w(\alpha)=pw_H(\alpha)$,  for some $p>0$ and for all $ \alpha \in \mathbb{F}_q$, then $m_w=M_w=p$ and  $w_{(P,w,\pi)}(x)=p w_{(P,\pi)}(x)$ $\forall$ $x \in \mathbb{F}_q^N$. Thus, we have:
\begin{proposition}\label{pwpiChept-I-sbwpic dpwpi=dppi}
	Let $\mathcal{C} $  be a $(P,w,\pi)$-code of length $N $ over $\mathbb{F}_q$ with minimum distance $d_{(P,w,\pi)}(\mathcal{C})$. If $w$ is a weight on $\mathbb{F}_q$ such that $w(\alpha)=pw_H(\alpha)$ for some $p>0$, for all $ \alpha \in \mathbb{F}_q$, then $d_{(P,w,\pi)}(\mathcal{C})=p d_{(P,\pi)}(\mathcal{C})$.
\end{proposition}

\begin{theorem}[Singleton Bound]\label{pwpiChept-I-sbwpic w=pw}
	Let $\mathcal{C} $  be a $[N,k]$ $(P,w,\pi)$-code of length $N $ over $\mathbb{F}_q$ with minimum distance $d_{(P,w,\pi)}(\mathcal{C})$. Let $w$ be such that $w(\alpha)=pw_H(\alpha)$ for some positive integer $p >0$ and for all $ \alpha \in \mathbb{F}_q$.  Then    
	$d_{(P,w,\pi)}(\mathcal{C})-p \leq (n - \frac{k}{s})p$ and $d_{(P,\pi)}(\mathcal{C})-1 \leq n - \frac{k}{s}$.
\end{theorem}
\begin{proof}
	The proof is straightforward from Theorem \ref{pwpiChept-I-sbwpic} and Proposition \ref{pwpiChept-I-sbwpic dpwpi=dppi}.
\end{proof}
\begin{proposition}\label{pwpiChept-I-n-kPperfectMDs}
	Let $\mathcal{C}$ be an $[N,k]$ $(P,w,\pi)$-code of length $N$ over $\mathbb{F}_q$. Let $w$ be such that $w(\alpha)=pw_H(\alpha)$ for some positive integer $p >0$ and for all $ \alpha \in \mathbb{F}_q$. If $\mathcal{C}$ is  $(n-\frac{k}{s})M_w$-perfect  then $\mathcal{C}$ is MDS.
\end{proposition}
\begin{proof}
	Since  $w(\alpha)=pw_H(\alpha)$, $m_w=M_w=p$. If $\mathcal{C}$ is $(N-k)M_w$-perfect, then $d_{(P,w,\pi)}(\mathcal{C}) > (n-\frac{k}{s})M_w$ and $d_{(P,w,\pi)} (\mathcal{C})-p \geq (n-\frac{k}{s})M_w$.  Hence,  $\mathcal{C}$ is MDS by Theorem \ref{pwpiChept-I-sbwpic w=pw}.
\end{proof}
\begin{corollary}
	Let $\mathcal{C}$ be an $[n,k]$ $(P,w)$-code of length $n$ over $\mathbb{F}_q$. Let $w$ be a weight on $\mathbb{F}_q$ such that $w(\alpha)=pw_H(\alpha)$ for some positive integer $p >0$ and for all $ \alpha \in \mathbb{F}_q$.  If $\mathcal{C}$ is  $(n- k)M_w$-perfect  then $\mathcal{C}$ is MDS.
\end{corollary}
\subsection{Duality theorem}
\par In this section, we will study the duality theorem for an MDS $(P,w,\pi)$-code when all the blocks are of same length. 
\par Given a poset $ P = ([n],\preceq)$,  its dual poset $\tilde{P}= ([n], \tilde{\preceq})$ is defined with the same underlying set $[n]$ such that $i \tilde{\preceq} j $ in $ \tilde{P}$ if and only if $j \preceq i$ in  $ {P} $. As a result, the order ideals of $\tilde{P} $ are precisely the complements of the order ideals of $P$; that is, $\mathcal{I}(\tilde{P} ) = \{I^c : I \in \mathcal{I}({P}) \}$.
\par 
Let  $\mathcal{C}$ be an $[N,k]$ $(P,w,\pi)$-code of length $N$ over $\mathbb{F}_q$ and   $\mathcal{C}^\perp$ be its dual.  Suppose that $\mathcal{C}$ is an $I$-perfect $(P,w,\pi)$-code. From [ref. \cite{bkdnsr}, Theorem $4.4$], an $[N,k]$ code $\mathcal{C}$ is $I$-perfect in $(P,\pi)$-space if and only if $\mathcal{C}^\perp$ is $I^c$-perfect in $(\tilde{P},\pi)$-space. Thus, from Theorem \ref{pwpiChept-I-pwtoP perfect}, we have 
\begin{proposition}\label{pwpiChept-I-dualperfect}
	Let $P$ be a poset on $[n]$ and $\tilde{P}$ be its dual poset. A $(P,w,\pi)$-code $\mathcal{C}$ is $I$-perfect if and only if $\mathcal{C}^\perp$ is an $I^c$-perfect $(\tilde{P},w,\pi)$-code.
\end{proposition}
\begin{proposition}\label{pwpiChept-I-n-kM_wperfectMDs I^t=1}
	Let $\mathcal{C}$ be an $[N,k]$ $(P,w,\pi)$-code of length $N$ over $\mathbb{F}_q$. Let $I$ be the unique ideal of $P$ with cardinality  $n-\frac{k}{s}$. If $\mathcal{C}$ is  $(n-\frac{k}{s})M_w$-perfect then $\mathcal{C}$ is MDS.
\end{proposition}
\begin{proof}
	Since $\mathcal{C}$ is $(n-\frac{k}{s})M_w$-perfect, $d_{(P,w,\pi)}(\mathcal{C}) > (n-\frac{k}{s})M_w$. Thus, $d_{(P,w,\pi)} (\mathcal{C})-m_w \geq n-\frac{k}{s}$ by Theorem \ref{pwpiChept-I-J insdie I minusMaxI}.  Hence, $\mathcal{C}$ is MDS from Theorem \ref{pwpiChept-I-sbwpic}.
\end{proof}
\begin{theorem}\label{pwpiChept-I-equitheorem}
	Let $\mathcal{C}$ be an $[N,k]$ linear $(P,w,\pi)$-code of length $N$  over  $\mathbb{F}_q$. If $\mathcal{I}^{n-\frac{k}{s}}=\{I\}$ then the following statements are equivalent:
	\begin{enumerate}[label=(\roman*)]
		\item\label{pwpiChept-I-15-1} $\mathcal{C} $ is MDS with respect to $(P,w,\pi)$-metric
		\item\label{pwpiChept-I-15-2} $\mathcal{C}$ is $I$-perfect    
		\item\label{pwpiChept-I-15-3} $\mathcal{C} $ is MDS with respect to $(P,\pi)$-metric
		\item\label{pwpiChept-I-15-4}  $\mathcal{C} $ is $(n-\frac{k}{s})$-perfect with respect to $(P,\pi)$-metric
		\item\label{pwpiChept-I-15-5} $\mathcal{C} $ is $(n-\frac{k}{s})M_w$-perfect with respect to $(P,w,\pi)$-metric.
	\end{enumerate}
\end{theorem}
\begin{proof}
	\ref{pwpiChept-I-15-1} $\Leftrightarrow$ \ref{pwpiChept-I-15-2}: If $\mathcal{C} $ is MDS with respect to $(P,w,\pi)$-metric then $\mathcal{C}$ is $I$-perfect for all $I \in \mathcal{I}^{n-\frac{k}{s}}$ by Theorem \ref{pwpiChept-I-Iperfectmdscode}.   Conversely, if $\mathcal{C}$ is $I$-perfect, then $\mathcal{C}$ is $(n-\frac{k}{s})M_w$-perfect by Theorem \ref{pwpiChept-I-rperfect}. Hence, $\mathcal{C} $ is MDS with respect to $(P,w,\pi)$-metric by Proposition \ref{pwpiChept-I-n-kM_wperfectMDs I^t=1}. \ref{pwpiChept-I-15-2} $\Leftrightarrow$ \ref{pwpiChept-I-15-3}: It is straightforward from the Corollary \ref{pwpiChept-I-posetmdsiff}. \ref{pwpiChept-I-15-2} $\Leftrightarrow$ \ref{pwpiChept-I-15-4}: It is straightforward from Theorem \ref{pwpiChept-I-rperfect} and Corollary \ref{pwpiChept-I-rperfectcoro1}. \ref{pwpiChept-I-15-2} $\Leftrightarrow$ \ref{pwpiChept-I-15-5}: It is straightforward from Theorem \ref{pwpiChept-I-rperfect}.
\end{proof}
\begin{theorem}[Duality Theorem]\label{pwpiChept-I-duality}
	Let $I$ denote the unique ideal of $P$ with
	cardinality $n-\frac{k}{s}$ and let $\mathcal{C}$ be an $[N,k]$ $(P,w,\pi)$-code of length $N$ over $\mathbb{F}_q$.
	Let  $\tilde{P}$ be the dual poset of the poset $P$. Then the following statements are equivalent:	
	\begin{enumerate}[label=(\roman*)]
		\item $\mathcal{C}$ is an MDS   $(P,w,\pi)$-code.
		\item $\mathcal{C}$ is  $I$-perfect.
		\item $\mathcal{C}^\perp$ is $I^c$-perfect.
		\item $\mathcal{C}^\perp$ is an MDS $(\tilde{P},w,\pi)$-code.
	\end{enumerate} 
\end{theorem}
\begin{proof}
	The proof is straightforward from Proposition \ref{pwpiChept-I-dualperfect} and Theorem \ref{pwpiChept-I-equitheorem}.
\end{proof}
\par This generalizes the duality theorem in \cite{aks} which was obtained for the case $s=1$.
\par Let  $P$ and $P'$ be any two posets on $[n]$ with partial orders $\preceq_P$ and $\preceq_{P'}$, respectively on $[n]$. We say $P'$ is finer than $P$ if $i \preceq_P j $ implies $ i \preceq_{P'} j$. 
\begin{theorem}
	Let $P,P'$ be any two posets  on $[n]$ where $P'$ is finer than $P$. If $\mathcal{C}$ is an MDS $(P,w,\pi)$-code, then C is MDS $(P',w,\pi)$-code as well.
\end{theorem}
\begin{proof}
	Let $\mathcal{C}$ be an MDS $(P,w,\pi)$-code. Since $P'$ is finer than $P$, $d_{(P,w,\pi)}(\mathcal{C}) \leq d_{(P',w,\pi)}(\mathcal{C})$ so that  $  \big\lfloor \frac{d_{(P,w,\pi)} (\mathcal{C})-m_w}{ M_w} \big\rfloor \leq \big\lfloor \frac{d_{(P',w,\pi)} (\mathcal{C})-m_w}{M_w} \big\rfloor $. We have  $  \big\lfloor \frac{d_{(P,w,\pi)} (\mathcal{C})-m_w}{ M_w} \big\rfloor  = N- \lceil log_{q}|\mathcal{C}| \rceil$ (from Theorem \ref{pwpiChept-I-sbwpic}), and hence $ N- \lceil log_{q}|\mathcal{C}| \rceil \leq \big\lfloor \frac{d_{(P',w,\pi)} (\mathcal{C})-m_w}{M_w} \big\rfloor  \leq N- \lceil log_{q}|\mathcal{C}| \rceil$. Thus, C is an MDS $(P',w,\pi)$-code.
\end{proof}
\section{Codes on NRT weighted block spaces}\label{NRT-P,w,pi-space}
Throughout this section, we  consider $P=([n],\preceq)$ to be a chain. Then each ideal in $P$ has a unique maximal element.
We have either $i \preceq j$ or $j \preceq i$ for any $i,j \in [n]$. Thus, $\max\langle i \rangle = i$,  $|\langle i \rangle |= i$ and  $|\mathcal{I}^t| =1$ for each $1 \leq t \leq n$ so that $  B_{(P,w,\pi)}(x, t M_w) = B_I(x)$. Let $v \in \mathbb{F}_{q}^N$, $I_v^{P,\pi}= \langle supp_{\pi}(v) \rangle$ and $j$ be the maximal element of  $I_v^{P,\pi}$. Then the $(P,w,\pi)$-weight of $v$ is
\begin{equation*}
	w_{(P,w,\pi)}(v)=  \tilde{w}^{k_j}{(v_j)} + (|I_v^{P,\pi}|-1) M_w.
\end{equation*}   
where $v_j \in \mathbb{F}_{q}^{k_j}$ and $\tilde{w}^{k_j}{(v_j)}=\max\{w(v_{j_t}) : 1 \leq t \leq k_j\}$.
This weight can be called as NRT weighted block weight or  $(P_{NRT},w,\pi)$-weight $	w_{(P_{NRT},w,\pi)}$ and the corresponding $(P,w,\pi)$-metric $d_{(P,w,\pi)}$ is called as NRT weighted block metric or  $(P_{NRT},w,\pi)$-metric $d_{(P_{NRT},w,\pi)}$ on $\mathbb{F}_{q}^N$. We call the  space $ (\mathbb{F}_{q}^N, ~d_{(P_{NRT},w,\pi)} )$ as an \textit{NRT weighted block space} or the \textit{NRT $(P,w,\pi)$-space}.
\par If $w$ is the Hamming weight on $\mathbb{F}_q$ then  NRT weighted block metric becomes \textit{NRT block metric} of \cite{nrt block}. Hence, the $(P,w,\pi)$-space with $P$ as a chain generalizes the NRT block space described in \cite{nrt block}.
\subsection{Weight Distribution of NRT Weighted Block Spaces}
 Let $A_r = \{x= x_{1} \oplus x_{2} \oplus \ldots \oplus x_{n} \in \mathbb{F}_{q}^{N}  : w_{(P_{NRT},w,\pi)}(x)=r  \}$.  Now, for any $i \in [n]$, $ \langle i \rangle  = \{1,2,\ldots,i\}$, $Max \langle i \rangle = \{i\}$ and  $|\mathcal{I}_{1}^{i}|=1$. Note that    $|PRT_{i-1}[r]|=1$  for any $r \leq nM_w$  so that $ARG[b]=PRT_{i-1}[r]$ where $b \in PRT_{i-1} [r]$. If $|\Gamma_i^P |=n_i=1$ for each level in a hierarchical poset then  hierarchical poset is a nothing but a chain poset. Thus, from Theorem \ref{r=Mw Hierarchical}, we have
\begin{theorem}\label{r=Mw weight chain poset}
	In the NRT weighted block space $ (\mathbb{F}_{q}^N, ~d_{(P_{NRT},w,\pi)} )$, we have 
	$	| A_r |= q^{ k_{1} + k_2 + \cdots + k_t} |D_{a}^{k_{t+1}}| $  where $r= tM_{w}+ a  \text{ with } 0 \leq t \leq n-1 \text{ and }  0 < a \leq M_{w}$ (by setting $k_0=0$). Moreover, if $k_i =k$ $\forall$ $i\in [n] $, then $	| A_r |= q^{ k t} |D_{a}^{k}| $.
\end{theorem}
\begin{theorem}\label{nrtwpirball}
	In $ (\mathbb{F}_{q}^N, ~d_{(P_{NRT},w,\pi)} )$, for any $ x \in \mathbb{F}_{q}^N $ we have 
	\begin{enumerate}[label=(\roman*)]
		\item $
			| B_{r}(x) | = q^{ k_{1}  + \cdots + k_t} (1+|D_{1}^{k_{t+1}}|+|D_{2}^{k_{t+1}}|+\ldots+ |D_{a}^{k_{t+1}}|)$ 
		where $r= tM_{w}+ a $ with  $0 \leq t \leq n-1$ and $0 \leq a \leq M_{w}$ (by setting $k_0=0$).
		\item $| B_{tM_w}(x) | =  q^{ k_{1} + k_2 + \cdots +  k_t}$.
		\item if $k_i =k$ $\forall$ $i\in [n] $, then  
		$
			| B_{r}(x) | = q^{kt} (1+|D_{1}^{k}|+|D_{2}^{k}|+\ldots+|D_{a}^{k}|)$
			\item $| B_{tM_w}(x) | =  q^{  k t}$.
	\end{enumerate} 
\end{theorem} 
\begin{proof}
	As $	| B_{r}(x) | = (1 + \sum\limits_{i=1}^{tM_W} |A_i| ) + \sum\limits_{j=1}^{a} |A_{tM_w+j}| $ and $1 + \sum\limits_{i=1}^{tM_w} |A_i| = q^{k_{1} + k_2 + \cdots + k_{t}}$, the result follows immediately.
\end{proof}
\subsection{ NRT block space}
When $w$ is the Hamming weight on $\mathbb{F}_q$, the NRT weighted block metric becomes the NRT block metric of \cite{nrt block}.
Here, $M_w=1$ and the maximum NRT block weight of any $ x \in \mathbb{F}_{q}^N $ is $n$. 
Moreover, $\tilde{w}^{k_i}(x_i)=1$  for any $x_i \in \mathbb{F}_q^{k_i}\setminus \{0\}$ and hence $|D_i^{k_i}|=(q-1)^{k_i}$.
Thus, for $1\leq r \leq n$, we have  $ |	A_r |=
q^{ k_{1} + k_2 + \cdots + k_{r-1}} (q-1)^{k_r} $ and $| B_{r}(x) | = q^{ k_{1} + k_2 + \cdots + k_r}$ with $k_0=0$. Further, when $k_i=1 ~\forall ~i \in [n]$, $| A_{r} | =  q^{r-1} (q-1)^r$ and $| B_{r}(x) | =  q^{r}$.
\subsection{Duality theorem}
\begin{theorem} \label{pwpiChept-I-rperfect chain}
	Let $\mathcal{C}$ be an $[N,k]$ linear $(P,w,\pi)$-code over  $\mathbb{F}_q$ and let $P$ be a chain. Then $\mathcal{C}$ is  $I$-perfect for  $I \in \mathcal{I}^t$ iff
	$\mathcal{C}$ is $tM_w$-perfect.
\end{theorem}
\begin{proof}
	Proof is straightforward from Theorem \ref{pwpiChept-I-rperfect}.
\end{proof}              
\begin{theorem}\label{pwpiChept-I-equitheoremchain}
	Let $\mathcal{C}$ be an $[N,k]$ $(P,w,\pi)$-code of length $N$ over $\mathbb{F}_q$ where $P$ being a chain. If $k_i=s$ $\forall$ $i \in [n]$.  Then the following statements are equivalent:
	\begin{enumerate}[label=(\roman*)]
		\item\label{pwpiChept-I-15-1} $\mathcal{C} $ is MDS with respect to $(P,w,\pi)$-metric.
		\item\label{pwpiChept-I-15-2} $\mathcal{C}$ is $I$-perfect for $I \in \mathcal{I}^{n-\frac{k}{s}}$. 
		\item\label{pwpiChept-I-15-3} $\mathcal{C} $ is MDS with respect to $(P,\pi)$-metric.
		\item\label{pwpiChept-I-15-4}  $\mathcal{C} $ is $(n-\frac{k}{s})$-perfect with respect to $(P,\pi)$-metric.
		\item\label{pwpiChept-I-15-5} $\mathcal{C} $ is $(n-\frac{k}{s})M_w$-perfect with respect to $(P,w,\pi)$-metric.
	\end{enumerate}
\end{theorem}
\begin{proof}
	Proof is straightforward from Theorem \ref{pwpiChept-I-equitheorem}.
\end{proof}
\begin{theorem}[Duality Theorem]\label{pwpiChept-I-duality chain}
	Let $P$ be a chain and $I$ be the ideal of $P$ with
	cardinality $n-\frac{k}{s}$. Let $\mathcal{C}$ be an $[N,k]$ $(P,w,\pi)$-code of length $N$ over $\mathbb{F}_q$.
	Let  $\tilde{P}$ be the dual poset of the poset $P$. Then the following statements are equivalent:	
	\begin{enumerate}[label=(\roman*)]
		\item $\mathcal{C}$ is an MDS   $(P,w,\pi)$-code.
		\item $\mathcal{C}$ is  $I$-perfect.
		\item $\mathcal{C}^\perp$ is $I^c$-perfect.
		\item $\mathcal{C}^\perp$ is an MDS $(\tilde{P},w,\pi)$-code.
	\end{enumerate} 
\end{theorem}
\begin{proof}
	The proof is straightforward from Theorem \ref{pwpiChept-I-duality}.
\end{proof}
\subsection{Weight distribution of an MDS $(P,w,\pi)$-code}
 Now, we determine the weight distribution of an MDS $(P,w,\pi)$-code of length $N$ over $ \mathbb{F}_q$ when $P$ is a chain and $k_i =s $ $\forall$ $i \in [n]$. To do this, first we analyse the distribution of codewords of $\mathcal{C}$ among $I$-balls. For that we define the following:
\par
Let $\mathcal{C} \subseteq \mathbb{F}_q^N $ be an MDS $(P,w,\pi)$-code. Let $A_r (\mathcal{C})= \{x \in \mathcal{C} : w_{(P,w,\pi)}(x) = r \}$ be the collection of codewords in $\mathcal{C}$ of $(P,w,\pi)$-weight $r$ where $0 \leq r \leq n M_w $. The following result describes how the codewords of a linear MDS $(P,w,\pi)$-code are distributed among the $r$-balls when all blocks have the same length:
\begin{proposition} \label{pwpiChept-I-mdsball}
	Let  $\mathcal{C} $ be an $[N,k]$ MDS $(P,w,\pi)$-code of length $N$ over $\mathbb{F}_q$ where $P$ is a chain and $k_i=s$ $\forall$ $i \in [n]$. Then, the number of codewords of $\mathcal{C}$ in an $r$-ball is given as follows:
	\begin{enumerate}[label=(\roman*)]
		\item  If $r \leq M_w (n-\frac{k}{s})$, then $|B_{(P,w,\pi)}(0,r) \cap \mathcal{C}| = 1$.
		\item  If $r =  (t+1) M_w$ where $t \geq n-\frac{k}{s}$, then $|B_{(P,w,\pi)}(0,r) \cap \mathcal{C} | = q^{k-s(n-t-1)  } $
		\item  If $r =  t M_w + \ell$ where $t \geq n-\frac{k}{s}$ and $1 \leq \ell < M_w$, then 
		\begin{align*}
			|B_{(P,w,\pi)}(0,r) \cap \mathcal{C} | = (1+|D_{1}^{s}| +\ldots+|D_{\ell}^{s}| ) q^{k-s(n-t)}.
		\end{align*}
	\end{enumerate}
\end{proposition}
\begin{proof}
	\sloppy	Since $\mathcal{C} $ is MDS, from Theorem \ref{pwpiChept-I-Iperfectmdscode}, $\mathcal{C}$ is $J$-perfect for $J \in \mathcal{I}^{n-\frac{k}{s}}$ and $|d_{(P,w,\pi)}(\mathcal{C})| > M_w(n - \frac{k}{s}) $.  
	Thus, when $r \leq M_w (n-\frac{k}{s})$, $B_{(P,w,\pi)}(0,r) \cap \mathcal{C} = \{0\}$.  
	\par  Consider $r > M_w (n-\frac{k}{s})$. Let $r = tM_w +\ell$  where $t \geq n-\frac{k}{s}$ and $\ell \in \{1,2,\ldots,M_w\}$. As $P$ is a chain, there exists a unique $ I  \in  \mathcal{I}^{t+1}$ such that $J \subseteq I$, and $B_J=B_{(P,\pi)}(0,n-\frac{k}{s})=B_{(P,w,\pi)}(0,(n-\frac{k}{s})M_w)$. Here, two cases arise: \\
	Case (a): If $ r = (t+1) M_w$, $t \geq n-\frac{k}{s}$, we have $B_{(P,w,\pi)}(0,r) = B_I$.  As $B_I$ is a linear  subspace of $\mathbb{F}_q^N$ and $B_{J}$ is a subspace of $B_I$, the number of cosets of $B_J$ in $B_I$ is $q^{s|I|-s|J|}$. Since $\mathcal{C}$ is $J$-perfect for every $J \in \mathcal{I}^{n-\frac{k}{s}}$,  every coset of $B_J$ in $\mathbb{F}_q^N$  contains exactly one codeword of  $\mathcal{C}$. Hence, $B_I$ contains $q^{s(t +1) - (ns - k) }$ codewords of  $\mathcal{C}$. \\ 
	Case (b): Here, $r =  tM_w +\ell$, $t \geq n-\frac{k}{s}$, $ 1 \leq \ell  < M_w$. Now, to get $|B_{(P,w,\pi)}(0,r) \cap \mathcal{C}|$, we need the number of disjoint translates of $B_{(P,w,\pi)}(0,(n-\frac{k}{s})M_w)$ in $B_{(P,w,\pi)}(0,r)$ whose union covers $B_{(P,w,\pi)}(0,r)$. Let $K = I \setminus J$. Now  $K$ is not an ideal of $P$ but a subset of $[n]$. Then, $|B_K|=q^{s(t+1)-(ns-k)}$. 
	\sloppy Now, $|B_{(P,w,\pi)}(0,r-(n-\frac{k}{s})M_w)|= |B_{(P,w,\pi)}(0,(t-n+\frac{k}{s})M_w+\ell)| = 	 (1+|D_{1}^{s}|+\ldots+|D_{\ell}^{s}|)  q^{k+s(t-n)}$ by Theorem \ref{nrtwpirball}. The  translates $x+ B_{(P,w,\pi)}(0,(n-k)M_w)$ with $x \in B_K$ are disjoint if the weight of $(t+1)^{th}$ block of $x$ is at most $l$ (i.e if $\tilde{w}(x_{t+1})\leq \ell)$. The number of such $x$ is given by   $|B_{(P,w,\pi)}(0,r-(n-\frac{k}{s})M_w)|$ which is $(1+|D_{1}^{s}|+|D_{2}^{s}|+\ldots+|D_{\ell}^{s}|)  q^{k+s(t-n)}$. Hence, $ |B_{(P,w,\pi)}(0,r) \cap \mathcal{C}|  = (1+|D_{1}^{s}|+|D_{2}^{s}|+\ldots+|D_{\ell}^{s}| ) q^{k+s(t-n)}$.  
\end{proof}
\begin{theorem}\label{pwpiChept-I-wtditchain}
	\sloppy	Let  $\mathcal{C}  $ be an $[N,k]$ MDS $(P,w,\pi)$-code of length $n$ over $\mathbb{F}_q$ with the minimum distance $d_{(P,w,\pi)}(\mathcal{C})$. Let $P$ be a chain and $k_i=s$ $\forall$ $i \in [n]$. Then 
	\begin{align*}
		|	A_r (\mathcal{C}) |= \begin{cases}
			1 &  \text{if} ~r = 0\\
			0 & \text{if} ~1 \leq r \leq d_{(P,w,\pi)}(\mathcal{C}) - 1  \\
			|D_{\ell}^{s}|  q^{k+s(t-n)}  & \text{if} ~r \geq d_{(P,w,\pi)}(\mathcal{C}) ~\text{and } r = tM_w + \ell \\ & \text{ where } t \geq n-k,   \ell \leq M_w.
		\end{cases} 	
	\end{align*}
\end{theorem}
\begin{proof}
	Clearly, $	|A_0 (\mathcal{C})| = 	1$ and $	|A_r (\mathcal{C})| = 	0$ if  $1 \leq r \leq d_{(P,w,\pi)}(\mathcal{C}) - 1$.
	As $\mathcal{C}$ is MDS, by Theorem \ref{pwpiChept-I-equitheoremchain},  $\mathcal{C}$ is $(n-\frac{k}{s})M_w$-perfect. Hence, $ d_{(P,w,\pi)}(\mathcal{C}) > (n-\frac{k}{s})M_w$. Now,
	$ |A_r (\mathcal{C})| = |B_{(P,w,\pi)}(0,r) \cap \mathcal{C}| - |B_{(P,w,\pi)}(0,r-1) \cap \mathcal{C}|$. 
	\par	 $\text{(i)}$: If $r =  (t+1) M_w \geq d_{(P,w,\pi)}(\mathcal{C})$, $t \geq n-\frac{k}{s}$, then $r-1 =  t M_w + M_w -1 $. Thus, from Proposition \ref{pwpiChept-I-mdsball}, we have 
	\begin{align*}
		|	A_r (\mathcal{C})| &= q^{k+s(t+1-n)} - (1+|D_{1}^{s}|+|D_{2}^{s}|+\ldots+|D_{M_w-1}^{s}| ) q^{k+s(t-n)}\\
		&= q^{k+s(t+1-n)} - ( q-|D_{M_w}|)^s q^{k+s(t-n)} \\
		&= (q^s-(q-|D_{M_w}|)^s)q^{k+s(t-n)}   \\
		&= |D_{M_w}^{s}|  q^{k+s(t-n)}.
	\end{align*} 
	\par  \sloppy $\text{(ii)}$: If $r= t M_w + 1 \geq d_{(P,w,\pi)}(\mathcal{C})$, $t \geq n-\frac{k}{s}$, then there exists an ideal $J \in \mathcal{I}^t$ such that  $|B_{(P,w,\pi)}(0,r-1)| = |B_{J}|$. From Proposition \ref{pwpiChept-I-mdsball}, we have 
	$		|A_r (\mathcal{C})|
	= 	(1 + |D_1 ^s|) q^{k+s(t-n)}- q^{k+s(t-n)} 
	=	|D_1^s| q^{k+s(t-n)}$.
	\par  $\text{(iii)}$: If $r =  t M_w +\ell \geq d_{(P,w,\pi)}(\mathcal{C}) $, $t \geq n-\frac{k}{s}$, $ 1 < l \leq M_w$, then  from Proposition \ref{pwpiChept-I-mdsball}, we have
	$	|A_r (\mathcal{C})| =  (1 + |D_1^s|+ \ldots+|D_{\ell}^s| ) q^{k+s(t-n)} - (1 + |D_1^s|+ \ldots+|D_{\ell-1}^s| ) q^{k+s(t-n)} 
	= |D_{\ell}^s| q^{k+s(t-n)}$.
	Hence proved.
\end{proof}


%
%

\end{document}